\documentclass[twoside]{article}
\usepackage[a4paper]{geometry}
\usepackage[latin1]{inputenc} 
\usepackage[T1]{fontenc} 
\usepackage{RR}
\usepackage{hyperref}

\usepackage{graphicx}
\usepackage{amsmath}
\usepackage{dsfont}
\usepackage{algorithm}
\usepackage{algorithmic}
\usepackage{amsthm}
\usepackage{amssymb}

\newcommand{\tuple}[1]{\mbox{$\langle #1 \rangle$}}
\newtheorem{thm}{Theorem}
\newtheorem{defn}[thm]{Definition}
\newtheorem{cor}[thm]{Corollary}
\newtheorem{exmp}[thm]{Example}
\newtheorem{rem}[thm]{Remark}
\newtheorem{lem}[thm]{Lemma}

\newtheorem{notation}[thm]{Notation}
\newtheorem{proposition}[thm]{Proposition}

\newcommand{\sO}{{\bf 0}}
\newcommand{\sI}{{\bf 1}}

\RRNo{7891}
\RRdate{February 2012}
\RRauthor{
Jean-Vivien Millo \and
Robert de Simone}
\authorhead{Millo \& de Simone}
\RRtitle{Ordonnancement p\'eriodique de graphes marqu\'es en utilisant les mots binaires balanc\'es}
\RRetitle{Periodic scheduling of marked graphs using balanced binary words}
\titlehead{Periodic scheduling of marked graphs using balanced binary words}
\RRresume{
Ce rapport presente un algorithme pour ordonnancer statiquement un graphe marqu\'e fortement connexe et vivant. L'algorithme propos\'e calcule la meilleur ex\'ecution pour laquelle le rendement effectif est maximal et la taille des places est minimale. L'agorithme propos\'e fournit les ordonnancements de chacun des noeuds de calcul sous la forme de mots binaires. Ces mots sont choisis balanc\'es. Les contributions du rapport sont à la fois l'algorithm propos\'e lui-m\^eme et la caract\'erisation de la meilleure ex\'ecution d'un graphe marqu\'e.
}
\RRabstract{
This report presents an algorithm to statically schedule live and strongly connected Marked Graphs (MG). The proposed algorithm computes the best execution where the execution rate is maximal and place sizes are minimal. The proposed algorithm provides transition schedules represented as binary words. These words are chosen to be balanced. The contributions of this paper is the proposed algorithm itself along with the characterization of the best execution of any MG.
}
\RRmotcle{graphe marqu\'e, ordonnancement, mot binaire balanc\'e}
\RRkeyword{Marked Graph, 
Scheduling,
Balanced binary word}
\RRprojets{AOSTE}
\RCSophia 

\begin{document}
\makeRR   
\tableofcontents

\section{Introduction}
In the System-on-Chip design domain, the trend is component based design. A new design is assembled from IP components which are interconnected through a network of point-to-point communication channels. In this area, the problem of long wire communication latency has emerged as a limitation \cite{Matzke97}. A channel is not able to forward a datum in a single step but requires many.

To solve this problem, a component based design has to be provided with its scheduling to take care of the latency issues. Luca Carloni {\em et al.} have proposed the theory of Latency Insensitive Design (LID) \cite{LCKMMASV} as a dynamic scheduling solution but LID is greedy in buffering element. From our initial tentative to improve the LID \cite{EURASIP06}, we have established that a component based design along with its latency issues can be modeled using Marked/Event graph (MG) \cite{CHEP_1971}. Consequently, from the challenge of scheduling a System-on-Chip design, we arrive to the more general and abstract challenge of scheduling an MG with respect to communication and computation latencies. 

To enter this challenge, we have developed the proposed algorithm which provide a statically computed execution to any live and strongly connected MG. The proposed algorithm can eventually be applied to any system (software, hardware, production chain) which can be abstracted as an MG with fixed communication and computation latencies.

It is clear from historical results \cite{QuadratCohenBaccelli92,CarlierChretienne,Ramchandani74} that a live MG always admits an execution irrespectively of the communication or computation latencies. The proposed algorithm consists in computing the {\em best} ASAP execution where execution rate is maximal and place sizes are minimal. These properties match with the requirements encountered in the domain of System-on-Chip design \cite{LCKMMASV}. Lastly, the proposed algorithm is extended to simply connected MGs. However, the validity of the computed execution relies on the on-demand availability of tokens on global inputs.


Except the proposed algorithm itself, the main contribution of the paper is the characterization of this {\em best} ASAP execution. From the initial marking, an guided execution shall lead to different markings. From each of these markings, the ASAP execution will be different and token accumulation in the places may vary. For example, in a given ASAP execution, a transition may fire all its tokens in sequence and then stall for the rest of the period, while in another ASAP execution, the same transition is fired every two instants. The first example promotes tokens accumulation. Within this set of ASAP executions, the one with the smallest tokens accumulation is called the {\em balanced ASAP execution}. This execution always exists and can be analytically computed for any MG. In a balanced ASAP execution, the binary words that represent the activities of the transitions through time ($\sI$ for activity, $\sO$ for inactivity) are all balanced. 


\paragraph{Related works}
Marked graphs is a well studied domain for more than forty years and many works are closely related to ours. \cite{deselEsparza} state the notions but also some results used in this paper. \cite{CarlierChretienne} and \cite{QuadratCohenBaccelli92} are the bases of our scheduling theory.

Historically, some works related to the notion of balancedness can be found in a publication of Jean Bernoulli in the 18$^{\rm th}$ century \cite{bernouilli}. Then they appeared as Christoffel words in the 19$^{\rm th}$ century \cite{christoffel}. More recently Christoffel words appear again in \cite{christoffelbylaurier}, and as Sturm words in \cite{sturm,sturmbyallauzen}, or as mechanical words in \cite{mechanicalword}. 
\cite{ALCOW} records the history of balanced binary words.

In \cite{Alanyali95analysisof,hAJEK1}, balanced binary words are used to balance load of Erlang network. In \cite{347482}, the authors try to minimize the data lose in a graph with fixed storage capacities by optimally routing data trough communication channels using balanced binary words.

\paragraph{Outline}
Section \ref{sec_overview} runs the proposed algorithm on an example.
Section \ref{section_MG} presents the MG definition followed by all the required results about static analysis of MG.
Balanced binary words are presented and studied in Section \ref{section_balanced_binary_word}. 
The proposed algorithm is presented in Section \ref{section_bal_stat_sched}, followed by the proofs of correctness and then
Section \ref{section_conclusion} discusses our results.

\section{Algorithm overview}
\label{sec_overview}
This section gives an informal overview of the major steps of the proposed algorithm. The vocabulary used is formally defined below. However it mostly refers to the usual and accepted definitions of the same in literature.

\paragraph{Algorithm inputs and outputs}
The proposed algorithm inputs are the live and strongly connected MG and its initial marking (initial token positions).
The proposed algorithm outputs are the computed execution and the size of the places for this execution.

\paragraph{Latency expansion and $\mathds{N}$-equalization}
In the MG presented in Figure \ref{figure_MGwithLatency}-a, the transition (rectangle) on the top has a computation latency of $1$. The right-most place (oval) has a communication latency of $3$. Usually, a token goes through a transition instantaneously and through a place in one step. When the computation latency is different from $0$, the tokens are kept for some time in the transition. Similarly, the tokens are kept longer in a place when its communication latency is more than $1$.

\begin{figure}[hbpt]
 \begin{center}
   {\includegraphics[scale=0.6]{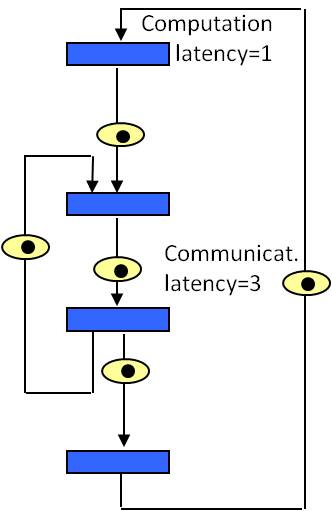}\,\,\,\,\,\,\,\,\,\,\,\,\,\,\,\,
    \includegraphics[scale=0.6]{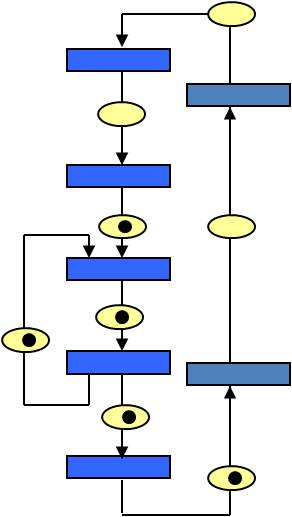}             \,\,\,\,\,\,\,\,\,\,\,\,\,\,\,\,
    \includegraphics[scale=0.6]{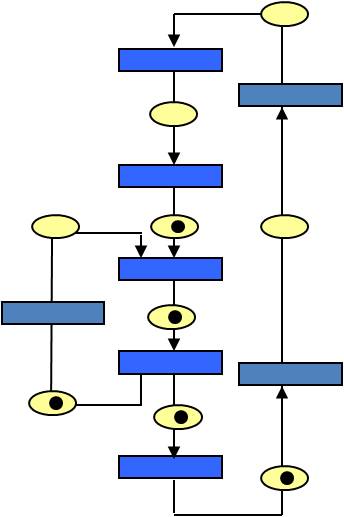}   \\ 
\begin{center}
a)\,\,\,\,\,\,\,\,\,\,\,\,\,\,\,\,\,\,\,\,\,\,\,\,\,\,\,\,\,\,\,\,\,\,\,\,\,\,\,\,\,\,\,\,\,\,\,\,\,\,\,\,\,\,\,\,\,\,\,\,\,\,\,
b)\,\,\,\,\,\,\,\,\,\,\,\,\,\,\,\,\,\,\,\,\,\,\,\,\,\,\,\,\,\,\,\,\,\,\,\,\,\,\,\,\,\,\,\,\,\,\,\,\,\,\,\,\,\,\,\,\,\,\,\,\,\,\,
c)
\end{center}
}
  \caption{a) an MG with a computation latency on the top-most transition and a communication latency on the right-most place.
  Its expansion gives the plain MG in b). The MG in c) is the $\mathds{N}$-equalized version of b).} 
  \label{figure_MGwithLatency}
  \end{center}
\end{figure}

In this representation, tokens evolution during the MG execution is not obvious. For example, in a place with a communication latency of $3$, some tokens could have been there for $1$ instant while others have been there for more than $3$ instants. The duration of their stay is not explicit.

To avoid this problem, the vertices with latencies are expanded in sequences of plain vertices such that the ``semantics of the latency" remains. A place with a communication latency $n$ is replaced by $n$ successive places while a transition with a computation latency $m$ is replaced by $m+1$ transitions interleaved with $m$ places. Thanks to this transformation, the exact location of tokens is known. Figure \ref{figure_MGwithLatency}-b is the expansion of Figure \ref{figure_MGwithLatency}-a.

In every ASAP executions reachable from the initial marking (after guided initialization), token accumulation mostly occurs in the same places. In the MG in Figure \ref{figure_MGwithLatency}-b, token accumulation occurs in the left-most place. When the accumulation is such that every token is kept at least $2$ instants in the place, the behavior of the place is similar to one with a communication latency of $2$. Thus it can be expanded. The {\em $\mathds{N}$-equalization} \cite{EURASIP06} detects these places analytically and increases their latencies accordingly. In Figure \ref{figure_MGwithLatency}-c, the MG is the $\mathds{N}$-equalized version of the one presented in Figure \ref{figure_MGwithLatency}-b. 

\paragraph{Running the proposed algorithm on an example}
The proposed algorithm is defined for an $\mathds{N}$-equalized MG where the latencies has been expanded. These steps are considered to be the preliminary steps of the proposed algorithm.

\begin{figure}[hbpt]
 \begin{center}
   {\includegraphics[scale=0.6]{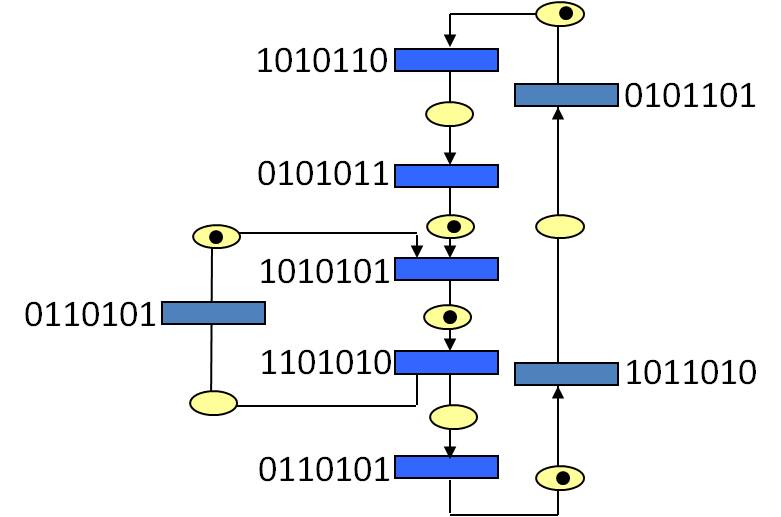}}
  \caption{The binary words associated to the transitions express the balanced ASAP execution of the MG introduced in Figure \ref{figure_MGwithLatency}-c.}
  \label{figure_balanced_sched}
  \end{center}
\end{figure}

Even in a $\mathds{N}$-equalized MG, token accumulation occurs. In some of the ASAP executions (reached after guided initialization), the accumulation is very limited while in others, many tokens can be regrouped in the same place. The balanced ASAP execution ($Exec_{periodic}$) has the lowest accumulation. Figure \ref{figure_balanced_sched} presents the schedule of every transitions according to $Exec_{periodic}$ ($\sO$ means inactivity, $\sI$ means activity). The first main step of the proposed algorithm computes $Exec_{periodic}$ analytically. 

In Figure \ref{figure_balanced_sched}, the marking from which $Exec_{periodic}$ occurs is called $M_{periodic}$. It is different from the initial marking ($M_0$) (Figure \ref{figure_MGwithLatency}-c). The second main step of the proposed algorithm computes $M_{periodic}$.

\begin{figure}[thbp]
 \begin{center}
   {\includegraphics[scale=0.6]{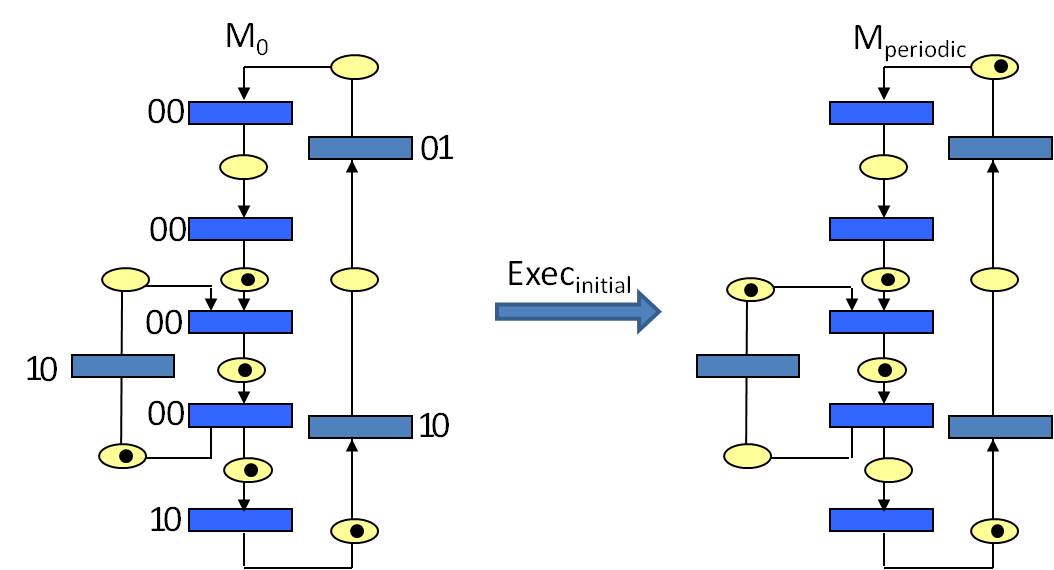}}
  \caption{From $M_0$ on the left to $M_{periodic}$ on the right following the initial guided execution.} 
  \label{figure_initialization2}
  \end{center}
\end{figure}

The third main step of the proposed algorithm consists in finding the guided initialization ($Exec_{initial}$) leading to $M_{periodic}$ from $M_0$. In Figure \ref{figure_initialization2}, The $2$-bits-length schedules attached to each transition is $Exec_{initial}$.

As one can see in Figure \ref{figure_final}, the computed execution is $Exec_{initial}$ followed by the infinite repetition ($\omega$) of $Exec_{periodic}$. It guarantees a maximal execution rate and a minimal accumulation of tokens. The proposed algorithm guarantees that place sizes are either $1$ or $2$. In the running example, every place size is $1$.

\begin{figure}[htbp]
 \begin{center}
   {\includegraphics[scale=0.6]{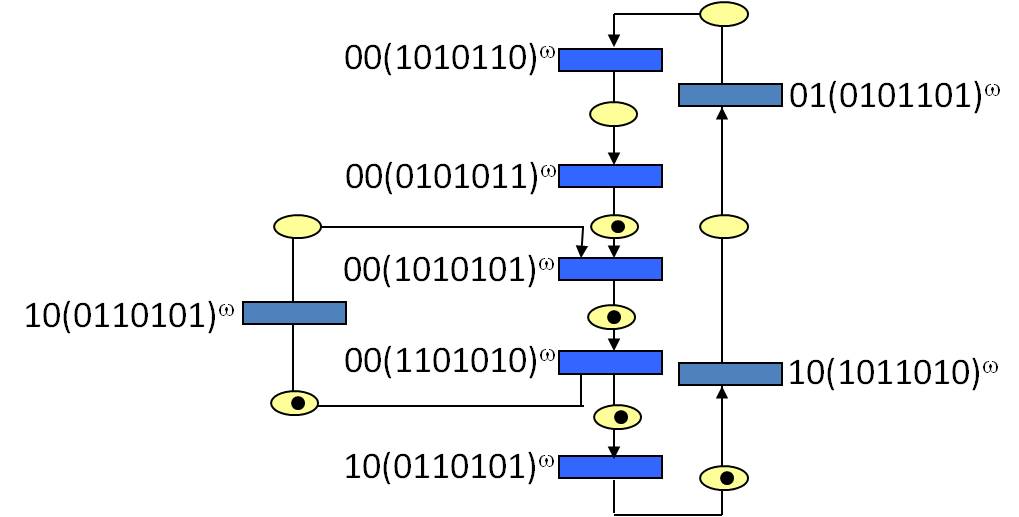}}
  \caption{The MG in $M_0$ and the execution computed by the proposed algorithm.} 
  \label{figure_final}
  \end{center}
\end{figure}

\section{Marked graph}
\label{section_MG}

This section presents the {\em Marked Graph} (MG) model also known as {\em Event Graph} along with classical definitions and results that will be used in the sequel. Our contributions in this section are the notion of delay presented in Section \ref{ssec_delay} and Theorem \ref{delayASAP}.

An MG is a graph where vertices can have two types: transitions and places. A place can stock tokens. The edges of a MG are called arcs. They cannot connect two vertices of the same type. 
A source is a transition without incoming arc. 
A sink   is a transition without outgoing arc. 

\begin{defn}[Marked Graph]
\label{def_mg}

A marked graph is a structure $G=\tuple{T,P,F,M_0}$ where
\begin{itemize}
\item{$T$ is a set of transitions.} 
\item{$P$ is a set of places.}
\item{$F\subseteq (T\times P)\cup (P\times T)$ is a set of arcs. If $t\in T$ and $p\in P$, $(t,p)$ and $(p,t)$ are two arcs resp. 
from $t$ to $p$ and from $p$ to $t$.}
\item{$M:P\to \mathds{N}$ is a marking. $M_0$ is its initial marking.}
\item{Each place has exactly one incoming and one outgoing arcs:
$\forall p\in P$, $|\{(t,p) \mid \forall t\in T\}|=|\{(p,t)\mid \forall t\in T\}|=1$.}
\end{itemize}
\end{defn}

The constraint on the number of place inputs and outputs guarantees that a token can be used by only one transition.
Consequently, the MG is said conflict free or deterministic. Figure \ref{figure_MGwithLatency}-b presents an MG with $7$ transitions (rectangles) and $8$ places (ovals). $5$ of these places contain one token (black dots).

\begin{notation}[Predecessor, successor]
Let $G$ be an MG, $t\in T$ and $p\in P$.
We note :
\begin{itemize}
\item ${}^\bullet t$ is the preset of $t$, ${}^\bullet t=\{p \mid (p,t) \in F\}$.
\item $t^\bullet$ is the postset $t$,  $t^\bullet=\{p \mid (t,p) \in F\}$.
\item ${}^\bullet p$ is the transition which precedes $p$, ${}^\bullet p=t$ such that $(t,p)\in F$.
\item $p^\bullet$ is the transition which succeeds $p$, $p^\bullet=t$ such that $(p,t)\in F$.
\end{itemize}
\end{notation}

\begin{defn}[Throughput of an MG, critical element]
\label{def_throughtput}
Let $G$ be an MG and $p$ be a place of $G$. A cycle $c$ is a path from $p$ to $p$. It is called elementary if all the transitions of the cycle are different. The marking of $c$ is $M(c)=\Sigma_{p\in c} M(p)$ and the latency of $c$, denoted $L(c)$, is the number of place on $c$. The value $M(c)/L(c)$ is the throughput of $c$. The cycle(s) with the lowest throughput is (are) said critical and the throughput of the MG is the one of the critical cycle(s). The transitions, arcs and places are said critical if they belong to a critical cycle.
\end{defn}

An MG is closed if it has neither source nor sink and 
it is connected if there exists a path, in the underlying undirected graph, relating any pair of vertices. 
It is strongly connected if there exists a path, in the MG itself, relating any pair of vertices. 
A strongly connected component (SCC) of an MG is a subgraph that is strongly connected (a subgraph of an MG is an MG composed of a subset of T, a subset of P, and a subset of F); it is said critical (CSCC) if all its elements are critical. 
A direct acyclic component (DAC) is a subgraph that does not contain any cycle.
In general, a connected MG is composed of DACs relating SCCs together.
A strongly connected MG is ever closed.

\subsection{Semantics of execution of an MG}
We define an execution semantics of an MG based on a logical time with a synchronous semantics. 
At the instant $0$, the MG is in its initial marking.
Then, an execution step leads to another marking at instant $1$ and so on. 
During a single execution step, many firable transitions can be fired simultaneously (synchronously) but a single transition can be fired only once.

\begin{defn}[Firable transition at a marking $M$ in an MG]
In an MG $G$, a transition $t\in T$ is firable at a marking $M$ if $\forall p\in {}^\bullet t$, $M(p)>0$. A source is always firable.
$F_M$ is the set of firable transitions at a marking $M$.
\end{defn}

\begin{defn}[MG execution model]
\label{def_execmodel}
Let $G$ be an MG and $M$ its current marking. 
An execution step is a transition relation from $M$ to $M'$ denoted $M \buildrel FT \over \longrightarrow M'$ with $FT\subseteq F_M$,
$\forall p\in P$, $M'(p)=M(p)+FT({}^\bullet p)-FT(p^\bullet)$. ($FT(t)=1$ if and only if $t\in FT$. $FT(t)=0$ otherwise).

An execution ($Exec$) of an MG is a finite or infinite sequence of execution steps:
$Exec = M_0 \buildrel FT_1 \over \longrightarrow M_1 \buildrel FT_2 \over \longrightarrow M_2 \buildrel FT_3 \over \longrightarrow... 
\buildrel FT_i \over \longrightarrow M_i \buildrel FT_{i+1} \over \longrightarrow ...$  where $FT_i \subseteq F_{M_{i-1}}$. 
\end{defn}

\begin{notation}[Concatenation of execution]
\label{def_compo_exec}
Let $G$ be an MG.
Let $Exec_0$ be a finite execution of $G$ from the marking $M_0$ to the marking $M_1$ and $Exec_1$ be a finite or infinite execution of $G$ from the marking $M_1$.

$Exec_0.Exec_1$ is the execution of $G$ formed by $Exec_0$ followed by $Exec_1$.
\end{notation}

\begin{notation}[ASAP and guided executions]
\label{notation_ASAP}
Let $G$ be an MG.
An execution of $G$ is said As Soon As Possible (ASAP) if and only if $\forall i$, $FT_i=F_{M_{i-1}}$ (all firable transitions are ever fired). 
An execution of $G$ is said guided if and only if $FT_i\subseteq F_{M_{i-1}}$.
In a guided execution, one has to decide which firable transitions are fired at every step.
\end{notation}

\begin{defn}[Scheduling and schedule]
\label{def_schedule}
Let $G$ be an MG with an execution $Exec$. Let $t\in T$ be a transition of $G$. The schedule of $t$ is the binary word relating the activity of $t$:
$Sched(t)=FT_1(t).FT_2(t)\cdots FT_i(t)\cdots$. The scheduling of $G$ for an execution $Exec$ is the mapping $t \to Sched(t)\mid \forall t\in T$.
\end{defn}

\begin{rem}[Scheduling and execution]
The successive steps of an execution can be deducted from its scheduling. Consequently, a scheduling defines an execution and vice versa.
\end{rem}

As we have seen in Section \ref{sec_overview}, the proposed algorithm computes an ASAP execution by computing the schedule of every transition.

\subsection{Classical results}
\label{ssection_clasresult}

\begin{defn}[Liveness]
\label{def_liveness}
An MG is live if there exists an execution where every transition is fired infinitely often.
\end{defn}

In \cite{CHEP_1971}, the authors show that the number of tokens on a cycle remains constant through execution. They deduce an MG is live
iff all its cycles contain at least one token.

\begin{defn}[Mutually reachable marking]
\label{def_reachmark}
Let $G$ be a strongly connected MG, $M$ and $M'$ two markings of $G$. $M$ and $M'$ are mutually reachable if there exists an execution sequence from $M$ to $M'$ and another from $M'$ to $M$.
\end{defn}

In \cite{deselEsparza}, the authors prove that two live markings, $M$ and $M'$, of the same strongly connected MG, $G$, are mutually reachable (through a guided execution) if for every cycle $c$ of $G$, $M(c)=M'(c)$. 

As we have seen in Section \ref{sec_overview}, the proposed algorithm computes an execution formed by an initial part followed by a steady part.
The steady part is not reachable from the initial marking through an ASAP execution.
Thus the initial part is a finite guided execution from the initial marking to
the first marking of the steady part. 
This operation is possible because the two markings are mutually reachable.

\subsection{Execution rate}

In \cite{CarlierChretienne}, the authors prove that the ASAP execution of a live and strongly connected MG is ultimately repetitive following an execution pattern. Equation \ref{figure_poele} shows the evolution of the marking of a live and strongly connected MG. $M_0$ is the initial marking and the arrows are ASAP execution steps.

\begin{equation}
   {\includegraphics[scale=0.15]{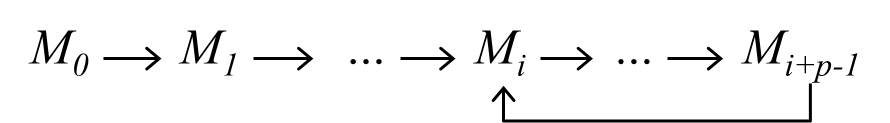}}
  \label{figure_poele}
\end{equation}

The period of the pattern is $\mathsf{p}$ and the number of firings of every transition within a period is $\mathsf{k}$ (the periodicity). We say that, the execution of the MG is $\mathsf{k}$-periodic with a period $\mathsf{p}$. In other words, the {\em execution rate} is $\mathsf{k}/\mathsf{p}$. In \cite{QuadratCohenBaccelli92}, the authors give a formula to calculate the exact value of the periodicity ($\mathsf{k}$) and the period ($\mathsf{p}$) of the ASAP execution of a closed MG. According to this formula, the execution rate ($\mathsf{k}/\mathsf{p}$) equals the value of the throughput given in Definition \ref{def_throughtput}. Thus the throughput is the maximal execution rate of the MG since it is the one of the ASAP executions. This is the one guarantied by the proposed algorithm.

\begin{proposition}[Maximal execution rate of an MG]
\label{prop_maxrate}
The maximal reachable execution rate is achieved by the ASAP execution of the MG.
\end{proposition}

\begin{rem}[Execution rate in an open MG]
\label{rem_thpopen}
The result on the maximal execution rate is valid for strongly connected MGs (and thus closed). In a simply connected open MG, the execution rate depends upon the execution rate of the source(s) but
the maximal execution rate is bounded by the worst throughput of its SCCs and can be calculated using the same formula given in \cite{QuadratCohenBaccelli92}. Consequently, if the source(s) fire(s) on demand, the MG can be considered closed.
\end{rem}

In Section \ref{ssec_simpyconnected}, the proposed algorithm is extended to simply connected MG.
In such a case, the proposed algorithm returns the schedules of the sources and sinks. 
The schedule of a source says when a token has to be generated by the source in order to feed the next transition and ensure the overall consistency of the execution.

\subsection{Size of places and boundedness}
As we have seen in Section \ref{sec_overview}, the proposed algorithm computes an execution which implies a minimal size of places. Let us now define this notion.

\begin{defn}[Size of places]
\label{def_capa}
Let $G$ be an MG, $Exec$ an execution and $p$ a place of $G$.
The size of $p$ on $Exec$ denoted $C_{Exec}(p)$ is the highest marking of $p$ during the entire execution.
\end{defn}

From the initial marking of a strongly connected MG, a guided execution can lead to any of the reachable markings. From each of these markings, there exists a bounded ASAP execution. The size of the places for these executions may vary. An execution has a minimal size of places if every places has a minimal size compared to the other ASAP executions.

\begin{defn}[Minimal size of places]
Let $G$ be an MG and $M_0$ its initial marking.
Let $\mathds{M}$ be the set of markings reachable from $M_0$.
Let $\mathds{E}$ be the set of ASAP executions from the markings of $\mathds{M}$.

$Exec\in \mathds{E}$ has a minimal size of places iff $\forall Exec' \in \mathds{E}$, $\forall p\in G$,
we have $C_{Exec}(p)\le C_{Exec'}(p)$.
\end{defn}

As we have seen in Section \ref{sec_overview}, the proposed algorithm computes an ASAP execution where place sizes are minimal. The extension of the proposed algorithm to simply connected MGs is discussed in Section \ref{ssec_simpyconnected} but it suffers from limitations since the execution of a simply connected MG may not be bounded.

\begin{defn}[Boundedness]
\label{def_safety}
An execution is bounded if the size of every place is bounded.
An MG is bounded if every execution is bounded.
\end{defn}

Whereas any SCC is bounded, the sizes of the places of a DAC are not. Let us assume an MG composed of two SCCs connected through a DAC. If the throughput of the upper SCC is superior to the throughput of the underneath SCC, the ASAP execution of the MG will lead to an infinite accumulation of tokens in the DAC. On contrary, If the throughput of the upper SCC is inferior to the throughput of the underneath SCC, the upper SCC will limit the execution rate of the underneath SCC and the behavior of the MG during an ASAP execution will be ultimately repetitive.

\begin{proposition}[ASAP execution of a connected MG]
\label{pro_asapnotbounded}
The ASAP execution of a connected MG may be unbounded.
\end{proposition}

In every bounded execution of a connected MG, all SCCs have the same execution rate. 
Consequently, the highest reachable execution rate is equals to the worst throughput among the SCCs. An execution at this rate will be ASAP for the SCCs with the worst throughput. The execution will also be ASAP for the underneath SCCs. However, the upper SCCs will be slow down to avoid accumulation and thus will not have an ASAP behavior.

\begin{proposition}[Bounded execution of a connected MG]
\label{pro_boundednotasap}
The bounded execution of a connected MG may not be ASAP.
\end{proposition}

The propositions \ref{pro_asapnotbounded} and \ref{pro_boundednotasap} explain why the proposed algorithm is restricted to strongly connected MG. However, this restriction can be abolished as discussed in section \ref{ssec_simpyconnected}. A simply connected MG can be transformed in a strongly connected one (by relating SCCs together) so that the proposed algorithm is applicable.

\subsection{Delay}
\label{ssec_delay}
During the execution of an MG, at a given instant, if a token reaches a place $p$ and is not consumed by $p^\bullet$ at the next instant, then the token 
is said delayed. This can happen in two cases: 1) when $p^\bullet$ is not fired; all the tokens in $p$ are delayed.
2) When $p^\bullet$ is fired and $p$ contains many tokens; all tokens in $p$ excepted the used one are delayed.
Globally, delays can be seen as a way for the MG to synchronize its branches together. A non critical cycle leans to take advance over critical cycles (it is faster) but eventually, the execution rate is the same for every one. So the delays reduce the execution rate of fast cycles to the execution rate of the slowest dynamically. The value $M{c_2}*L{c_1}-M{c_1}*L{c_2}$ represents the number of delay required during a period of execution to synchronize the cycle $c_1$ with the cycle $c_2$.

\begin{defn}[Delay]
\label{def_delay}
Let $G$ be an MG and $p\in P$ be a place of $G$. Let $Exec$ be an execution of $G$. Following the notation of Definition \ref{def_execmodel},
$Delay(p,i)$ is the number of delays occurring in $p$ at the $i^{th}$ step of $Exec$. 
$Delay(p,i)=M_{i-1}(p)-FT_i(p^\bullet)$. 
\end{defn}

After the initial part, the sum of delays in the places of a cycle during a period of execution reflects the difference of rate between a critical cycle and the current cycle.

\begin{thm}[Delay in a cycle during a period of execution]
\label{th_delay_cycle}
Let $G$ be an MG and $c$ a cycle of $G$. Let $Exec$ be an [ultimately] $\mathsf{k}$-periodic execution of $G$ with a period $\mathsf{p}$. 
Let $j_0$ be an upper bound of the length of the initialization.
\begin{equation*}
\Sigma_{p \in c}\Sigma_{i=1}^{\mathsf{p}}Delay(p,j_0+i)=M(c)*\mathsf{p}-L(c)*\mathsf{k}
\end{equation*}
\end{thm}
\begin{proof}
In $c$, at each instant, $M(c)$ tokens are present.
This means $M(c)*\mathsf{p}$ transitions could be fired over a period of execution. 
However, in a period of execution, every transition of $G$ are fired $\mathsf{k}$ times. 
This means $L(c)*\mathsf{k}$ transitions are effectively fired on $c$ during a period of execution. 
The difference between the amount of possible fired transition and the amount of effective fired transition 
is the number of delays in $c$ over a period.
\end{proof}

The spatial distribution of delays is the exact location where the delays occur during a period of execution.

\begin{defn}[Spatial distribution of delays]
\label{def_sdd}
Let $G$ be an MG. 
Let $Exec$ be an [ultimately] $\mathsf{k}$-periodic execution of $G$ with a period $\mathsf{p}$.
$D: P \to \mathds{N}$ is called a spatial distribution of delays 
if $\forall c$, the cycles of $G$, $\Sigma_{p\in c} D(p)=M(c)*\mathsf{p}-L(c)*\mathsf{k}$.

$Exec$ is said to be {\em based on $D$} if after the initial part, the delays in $exec$ during a period of execution occur as expressed in $D$.
\end{defn}

In the specific case of an ASAP execution, the delays occur as late as possible in the MG. This makes the corresponding spatial distribution of the delays unique for a given strongly connected MG. This spatial distribution is called the ``latest delays position". The theorems \ref{thm_ldp_ex} and \ref{delayASAP} prove these claims.

\begin{defn}[Latest delays position]
\label{def_latestdelaypos}
Let $G$ be a strongly connected MG. 
Let $D$ be a spatial distribution of the delays in $G$.
$D$ is the latest delays position if for all transition $t$ of $G$, there exists at least one place $p$ in ${}^\bullet t$ such that $D(p)=0$.
\end{defn}

\begin{thm}[Existence of the latest delay position]
\label{thm_ldp_ex}
Let $G$ be a strongly connected MG with a throughput inferior or equal to $1$.
The latest delay position ever exists for $G$.
\end{thm}
\begin{proof}
A spatial distribution of delays $D$ can be deducted from a period of the ASAP execution of $G$.
$\forall p\in P$, $D(p)=\Sigma_{i=1}^{\mathsf{p}}Delay(p,j_0+i)$ where $j_0$ is the length of the initial part.

Either $D$ is the latest delay position or there exists at least a transition $t$ for which every places in the preset of $t$ has at least $n$ delays (with $n\ge 1$).
In the second case, $n$ delays can be removed to every place in the preset of $t$ and added to every place in the postset of $t$.
This transformation gives another (valid) spatial distribution of delays for which $t$ has at least one place in its preset without delay.

The iteration of this transformation reaches a fix point because no delay appends on the critical cycle. The fix point is the latest delay position. One should note the similarity of this argument to the liveness condition. 
\end{proof}

\begin{thm}[Latest delay position and ASAP execution]
\label{delayASAP}
Let $G$ be a strongly connected MG with an execution $Exec$. 
$Exec$ is based on the spatial distribution of delays $D$.
i) If $D$ is the latest delays position, then $Exec$ is ASAP. 
ii) Let $Exec'$ be another ASAP execution of $G$ from another initial marking $M_0'$. 
$Exec'$ is based on the spatial distribution of delays $D'$.
If $M_0$ and $M_0'$ are mutually reachable, then $D=D'$.
\end{thm}
\begin{proof}
i) If $\forall t\in T$, $\exists p \in\{{}^\bullet t\}$ such that $D(p)=0$, as soon as $M(p)>0$, $t$ fires. This is a ASAP execution.\\
ii) If $M_0$ and $M_0'$ are mutually reachable, the number of tokens per cycle is the same in $M_0$ and $M_0'$ for every cycle of $G$ \cite{deselEsparza}. Consequently, the number of delays per cycle is the same in $D$ and $D'$.

Now let us assume there exists a place $p$ such that $D(p)\neq D'(p)$. Let $path_1$ and $path_2$ be two paths in the graphs.
$path_1$ goes from a transition of a critical cycle to ${}^\bullet p$ and $path_2$ goes from $p^\bullet$ to a transition of a critical cycle. We assume without lost of generality that the number of delays on $path_1$ is the same according to $D$ and $D'$.
$path_1$ followed by $p$ followed by $path_2$ followed by a section of a critical cycle forms a cycle for which the number of delays is the same according to $D$ and $D'$. Since $D(p)\neq D'(p)$, the number of delays on $path_2$ is different on $D$ and $D'$.

Since $D$ and $D'$ are the latest delays position, there exists a path $path_0$ from a transition of a critical cycle to $p^\bullet$ which do not contains any delay (the construction of this path can be done by backtracking from $p$: while reaching a transition, the input place without delay is selected, a critical cycle will ultimately be reached). But $path_0$ followed by $path_2$ followed by a section of the same critical cycle forms a cycle where the number of delays is different according to $D$ and $D'$. $M_0$ and $M_0'$ are not mutually reachable.
\end{proof}

As we have seen in Section \ref{sec_overview}, the proposed algorithm computes an ASAP execution. 
This ASAP execution is based on the latest delays position of $G$. 
In some sense, the proposed algorithm proves that there ever exists an ASAP execution based on the 
latest delay position.

\subsection{Latencies}
\label{ssec_latexp}
The preliminary step of the proposed algorithm is the expansion of the vertices with latency in plain vertices.

\begin{defn}[MG with communication/ computation latencies]\-\\
Let $G$ be an MG. A marked graph with latency $G'$ is a tuple $\tuple{G,L_{com},L_{cal}}$:
\begin{itemize}
\item The mapping $L_{com}$:$P\to \mathds{N}\backslash \{0\}$ gives the communication latencies of places.
\item The mapping $L_{cal}$:$T\to \mathds{N}$ gives the computation latency of transitions ($cal$ stands for calculation).
\end{itemize}
\end{defn}

A place with a communication latency of $n$ keeps every token at least $n$ instants. A transition with a computation latency of $m$ keeps every token exactly $m$ instants. According to Definition \ref{def_execmodel}, the latency of a transition in a plain MG is $0$ and the latency of a place is $1$. The tokens go through transitions instantaneously but stay at least one instant in a place. 
The transformation from an MG with latencies to an MG without latency has been introduced by Chander Ramchandani in \cite{Ramchandani74}.
This transformation preserves the semantics of a latency.


Figure \ref{figure_MGwithLatency}-a presents an MG with computation latencies on the top transition and communication latencies on the right-most place.
Figure \ref{figure_MGwithLatency}-b is the expansion of Figure \ref{figure_MGwithLatency}-a.
The top-most transition is replaced by two transitions with a place in between which represents the computation latency. 
The right-most place is replaced by three places interlaced by two transitions. Each of the three places represents a communication latency.

Liveness, closedness, (strongly) connection, throughput, execution rate, number of cycles, and number of tokens per cycle remain constant through the latency expansion process.

\subsection{$\mathds{N}$-equalization}
\label{subsec_eqn}
In an MG where a cycle $c$ is largly faster that the critical cycle,
any ASAP execution will lead to a situation where 
a place of $c$ will keep every token at least two instants.
In consequence, the behavior of this place is exactly the same as two places in sequence
with a dummy transition in-between. The $\mathds{N}$-equalization performs this transformation wherever it is required.
The MG in Figure \ref{figure_MGwithLatency}-c is the $\mathds{N}$-equalized version of the MG in Figure \ref{figure_MGwithLatency}-b.

The resulting $\mathds{N}$-equalized MG has the same behavior as the original one but
the throughput of $c$ has changed. It has been reduced to approach the 
critical one but cannot become less.
It may append that some non-critical cycles can become critical and the value of $\mathsf{k}$ and $\mathsf{p}$ can change but the ratio $\mathsf{k}/\mathsf{p}$ remains constant. The major expected change is that for every places in the resulting MG, the number of delays over a period becomes bounded by $\mathsf{k}$. More details about $\mathds{N}$-equalization is available in \cite{MEMOCODE06,EURASIP06}.

\begin{defn}[$\mathds{N}$-equalized MG]
\label{NQequalization}
An MG $G$ is said $\mathds{N}$-equalized if and only if
every transition belonging to a strongly connected component of $G$ belongs to a cycle $c$ such that: 
\begin{equation*}
M(c)/(L(c)+1) < throughput(G) \le throughput(c)
\end{equation*}
\end{defn}

\begin{lem}[Delay in a $\mathds{N}$-equalized MG]  
\label{lemma_nequa}
Let $G$ be a $\mathds{N}$-equalized MG. Let $p$ be a place of $G$. Let $D$ be a spatial repartition of delays.
$0\le D(p)<\mathsf{k}$ holds (where $\mathsf{k}$ is the periodicity of $G$).
\end{lem}

\begin{proof}
For all places $p$ in $G$, there is a cycle $c$ such that 
$\Sigma_{p \in c} D(p)=M(c)*\mathsf{p}-L(c)*\mathsf{k}$.

Moreover, if $G$ is $\mathds{N}$-equalized,\\
 $M(c)/(L(c)+1) < throughput(G) \le throughput(c)$\\ 
 $\Leftrightarrow M(c)/(L(c)+1)< \mathsf{k}/\mathsf{p} \le M(c)/L(c)$.
The two inequations hold:
\begin{itemize}
\item $M(c)/(L(c)+1)< \mathsf{k}/\mathsf{p} \Leftrightarrow M(c)*\mathsf{p}<(L(c)+1)*\mathsf{k} \Leftrightarrow M(c)*\mathsf{p}-L(c)*\mathsf{k}<\mathsf{k}$.
\item $\mathsf{k}/\mathsf{p} \le M(c)/L(c) \Leftrightarrow M(c)*\mathsf{p}-\mathsf{k}*L(c)\ge 0$.
\end{itemize}
The two inequations can be merged in $0 \le  M(c)*\mathsf{p}-L(c)*\mathsf{k} < \mathsf{k}  \Leftrightarrow 0 \le \Sigma_{p \in c} D(p) < \mathsf{k}$.
Even if all the delays of the cycle are merged in one place, $D(p)<\mathsf{k}$.
\end{proof}

The major complexity of the $\mathds{N}$-equalization comes from the interleaving of cycles in the MG. The addition of an extra place on a path may increase the latency of many cycles and some of them can become slower that a critical cycle while some others still require extra places. Consequently, all the cycles have to be considered simultaneously to find the correct location of the additional places. In \cite{MEMOCODE06,EURASIP06}, integer linear programming is used to specify all the $\mathds{N}$-equalization constraints. A more elegant solution can be built based on (max,plus) algebra \cite{QuadratCohenBaccelli92} by considering the incidence matrix of the MG and its evolution over a period.

In Figure \ref{figure_MGwithLatency}, 
the $\mathds{N}$-equalization may appear trivial because many places belong to only one cycle.
The left cycle in Figure \ref{figure_MGwithLatency}-b is faster than the right cycle, so an
extra place can be added after the leftmost place.
The critical (right) cycle has a throughput of $4/7$. The left cycle has a throughput $2/3$. 
The inequation $2/(3+1)<4/7\le 2/3$ holds.

\begin{figure}[hbpt]
 \begin{center}
   {\includegraphics[scale=0.75]{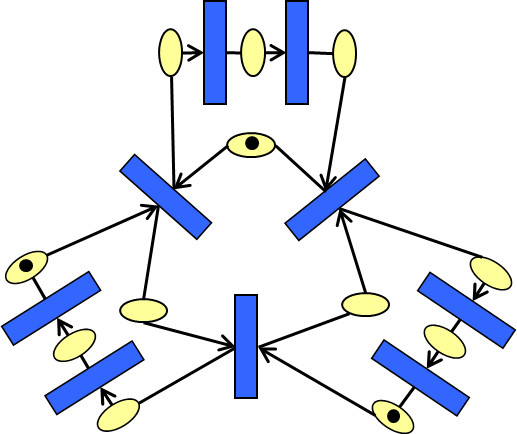}}
  \caption{This MG is already $\mathds{N}$-equalized.} 
  \label{figure_ntee}
  \end{center}
\end{figure}

Figure \ref{figure_ntee} presents a non-trivial example of $\mathds{N}$-equalization.
The outer cycle is critical with a throughput $2/9$. There is three cycles with a throughput
$1/4$. Since $1/(4+1) < 2/9\le 1/4$, the $\mathds{N}$-equalization condition of 
Definition \ref{NQequalization} hold. The inner cycle has a throughput $1/3$ so it
seems that an extra place could be added to equalize it ($1/(3+1)>2/9$)
but every place of this cycle also belongs to another cycle with a throughput $1/4$.
Consequently, the MG is already $\mathds{N}$-equalizated.

As we have seen in Section \ref{sec_overview}, the $\mathds{N}$-equalization of an MG is the preliminary step of the proposed algorithm.

\section{Balanced binary words}
\label{section_balanced_binary_word}

This section presents the basic definitions and well-known results on balanced binary words (\cite{ALCOW}). Up to our knowledge, Theorem \ref{theooperation}, that presents the relation between the operation of rotation and transposition, is original. The goal of this section is to present all these results in a way that eases the comprehension of the proposed algorithm.

\subsection{Finite and infinite binary words}
\label{subsection_bin_word}

As usual the set of binary values is noted $\mathds{B}= \{\sO, \sI\}$, 
$\mathds{B}^*$ the set of finite binary words, 
$\mathds{B}^n$ the set of binary words of length $n$,
$\mathds{B}^+$ the set of non-empty finite binary words, 
$\mathds{B}^{\omega}$ the set of infinite binary words, 
and  $\varepsilon$ the empty word.
We note  $\mathds{B}^{\infty}=\mathds{B}^*\cup \mathds{B}^{\omega}$,
the set of finite or infinite binary words.

For $u \in \mathds{B}^{\infty}$, we note $|u|$ the length of $u$ (with $|u|=\infty$ whenever $u\in \mathds{B}^{\omega}$).
Similarly we note $|u|_1$ and $|u|_0$ the number of occurrences of letters $\sI$ and $\sO$ in $u$ respectively.
Also, for $u \in \mathds{B}^+$ we note $slope(u)$ the ratio $|u|_1 / |u|$. 
$\mathds{B}_\mathsf{k}^\mathsf{p}=\{u\mid u\in \mathds{B}^\mathsf{p}$ and $|u|_1=\mathsf{k}\}$.
For $i\leq |u|$ we note $u(i)$ the $i^{th}$ letter of $u$.

The lexicographic ordering on words is defined as:
for $u, v \in \mathds{B}^{\infty}$, $u < v$ iff $\exists i\in \mathds{N}$, $\forall j<i$, $ u(j)=v(j)$ and 
either $u(i)=\sO$ and  $v(i)=\sI$ or $|u|=i-1$ and $|v|\ge i$. This order is total. For any finite subset $V$ of $\mathds{B}^\infty$, 
$inf(V)$ and $sup(V)$ are respectively its lowest and highest elements for this ordering.
Finally, for $u \in \mathds{B}^*$ and $v \in \mathds{B}^\infty$, $u$ is a factor of $v$ if  $\exists u_1 \in \mathds{B}^*$, 
$u_2 \in \mathds{B}^\infty$ such that $v=u_1.u.u_2$.


\begin{defn}[Ultimately $\mathsf{k}$-periodic binary word]
\label{def_ukpw}
An infinite binary word is called {\em ultimately $\mathsf{k}$-periodic} if it is of the form $u.v^\omega$, with $u\in\mathds{B}^*$ and 
$v \in \mathds{B}^+$ with $|v|_1=\mathsf{k}>0$.
\end{defn}
It is called simply {\em $\mathsf{k}$-periodic} if in addition $u=\varepsilon$.
It is called {\em ultimately periodic} if $\mathsf{k}=1$.
It is called only {\em periodic} if both conditions occur.
For an ultimately $\mathsf{k}$-periodic word, $u$ is called the {\em initial} part, $v$, the {\em steady} part, $\mathsf{k}=|v|_1$ is
the {\em periodicity}, and $\mathsf{p}=|v|$ is the {\em period}. By definition $slope(u.v^\omega)=slope(v)$. 
$\mathds{P}$ is the set of ultimately periodic infinite binary words and $\mathds{P}_\mathsf{k}^\mathsf{p}$ is the set of such word of periodicity $\mathsf{k}$ and period $\mathsf{p}$.

\begin{exmp}
$11.(0110101)^\omega$  is 4-periodic with period 7, and so is in $\mathds{P}_4^7$.
\end{exmp}

Because the ASAP execution of an MG is ultimately periodic, the proposed algorithm mainly focus on a single period of execution that aim to be indefinitely repeated. Thus, the following results concern finite binary words. In the proposed algorithm, for each transition $t$ of an MG, the appropriate words $v_t$ and $u_t$ are found and the ultimately $\mathsf{k}$-periodic word $u_t.(v_t)^\omega$ is built to represent the schedule of $t$.

\subsection{Rotation and transposition}
\label{subsection_rot_trans}

As we have seen in Section \ref{sec_overview}, the proposed algorithm computes the schedule of every transition of the MG.
To do so, the schedule of a transition is deducted from the schedule of one of its predecessors ($\{{}^\bullet p\mid p\in {}^\bullet t\}$) using the transposition and rotation. In Section \ref{subsection_relation}, we illustrate the link between the rotation and the effect of a latency on a schedule as well as the link between the transposition and the effect of a delay on a schedule.

\begin{defn}[Unitary forward rotation]
The unitary forward rotation is defined as $\rho$: $\mathds{B}^* \to \mathds{B}^*$,
$\rho(\varepsilon)=\varepsilon$, and
$\forall u \in \mathds{B}^*$, $\forall b \in \mathds{B}$, $\rho(u.b)=b.u$.
\end{defn}

\begin{defn}[Rotation]
\label{def_rotation}
Let $u \in \mathds{B}_\mathsf{k}^\mathsf{p}$. 
we note $\rho^n(u)$ the $n$ successive unitary forward rotation of $u$.
$\rho^{0}(u)=u$, $\rho^{1}(u)=\rho(u)$, $\rho^n(u)=\rho^{n-1} \circ \rho(u)$ and,\\
$\rho^{-n}(u)=v$ when $u=\rho^{n}(v)$. The parameter $n$ is called the spin of the rotation.
\end{defn}

\begin{exmp}
$\rho^3(1101010)= 0101101$, 
$\rho^{-3}(1101010)= 1010110$ and, 
$\rho^\mathsf{p}(u)=\rho^0(u)=u$
\end{exmp}

\begin{defn}[Orbit]
Let $u \in \mathds{B}^*$, the set of all rotations of $u$ is called the orbit of $u$ and is noted $O(u)$.
\end{defn}

\begin{exmp}
For $u=0110101$, $O(u)=\{u,\rho^1(u),...,\rho^6(u)\}=\{0110101,1011010,$ $0101101,1010110,0101011,1010101,1101010\}$
\end{exmp}

\begin{defn}[Transposition]
\label{transposition}
Let $u$, $v \in \mathds{B}^{\infty}$. $v$ is called the unitary forward transpose of $u$
(or simply transpose for short) and noted $v=\tau(u,\Delta)$, iff $\exists u_1 \in \mathds{B}^*$ and $\exists u_2\in \mathds{B}^{\infty}$,
$u=u_1.1.0.u_2$, $v=u_1.0.1.u_2$, and $\Delta=|u_1|+1$. $\Delta$ is called the location of the transposition.
By definition, if $u=0.u_1.1$, $\tau(u,|u|)=1.u_1.0$ where $u$ is finite.
\end{defn}

\begin{exmp}
$\tau(1010101,3)=1001101$, 
$\tau(1101010,3)$ is not defined, 
$\tau(011,3)=110$,
$\tau((10101)^{\omega},3)=10011.(10101)^{\omega}$ and,
$(\tau(10101,3))^{\omega}=(10011)^{\omega}$.
\end{exmp}

\subsection{Balanced binary words}
\label{subsection_defBBW}

The proposed algorithm computes an execution where all schedules are ultimately $\mathsf{k}$-periodic balanced binary words with a period $\mathsf{p}$.

\begin{defn}[Balanced binary word]
\label{balanced}
A finite binary word $u \in \mathds{B}^+$ is said {\em balanced} if $\forall v,t$, two factors of $u^\omega$ such that
$|v|=|t|$,  the following property holds: $-1\leq |v|_1-|t|_1\leq 1$.
\end{defn}

The set of finite balanced binary words with length $\mathsf{p}$ and containing $\mathsf{k}$ occurrences of $\sI$ is denoted by $\mathds{S}_\mathsf{k}^\mathsf{p}$. Also, $u \in \mathds{S}_\mathsf{k}^\mathsf{p}$ is said primitive when $\mathsf{k}$ and $\mathsf{p}$ are mutually prime. By extension an ultimately periodic word is called {\em balanced} if its steady part is. We have chosen the letter $\mathds{S}$
for {\em Smooth}.

In \cite{ALCOW}, the authors prove that i) in a balanced binary word $u$, the number of $\sI$ in every factor of $u^\omega$ with a length $l$ is either $\lfloor {l*|u|_1 /|u|}\rfloor$ or $\lceil {l*|u|_1 /|u|}\rceil$, ii) all the balanced binary words with the same slope are equivalent by rotation (let $u,v\in \mathds{S}_\mathsf{k}^\mathsf{p}$, $O(u)=O(v)=\mathds{S}_\mathsf{k}^\mathsf{p}$), iii) $inf(\mathds{S}_{\mathsf{k}}^{\mathsf{p}})=0.u.1$ and $sup(\mathds{S}_{\mathsf{k}}^{\mathsf{p}})=1.u.0$ ($u \in \mathds{B}^{\mathsf{p}-2}$), and lastly iv) whenever $\mathsf{k}$ and $\mathsf{p}$ are not mutually prime, every balanced binary word in $\mathds{S}_\mathsf{k}^\mathsf{p}$ (called in this case non-primitive) is the repetition of a smaller primitive balanced binary word: let $0< \mathsf{k} \leq \mathsf{p}$ and $GCD(\mathsf{k},\mathsf{p})=x$,
$\forall u \in \mathds{S}_\mathsf{k}^\mathsf{p}$, $\exists v \in \mathds{S}_{\mathsf{k}/x}^{\mathsf{p}/x}$ such that $u=v^x$.


When the proposed algorithm meets none-primitive balanced binary word, it considers the primitive balanced binary word imprinted into it. The execution is correct because when $u=v^x$, we have $u^\omega=v^{x^\omega}=v^\omega$.

\subsection{Transposition on balanced binary words}
\label{subsection_transpo}

Definition \ref{def_taufunction} defines a bijective function of transposition from $\mathds{S}_{\mathsf{k}}^{\mathsf{p}}$ to $\mathds{S}_{\mathsf{k}}^{\mathsf{p}}$. It requires some intermediate results.

\begin{lem}[Transposition in $\mathds{S}_\mathsf{k}^\mathsf{p}$]
\label{delayskp}
$\forall u\in \mathds{S}_\mathsf{k}^\mathsf{p}$ with $\mathsf{k}$ and $\mathsf{p}$ relatively prime, There exists a unique $\Delta$ such that $\tau(u,\Delta) \in \mathds{S}_\mathsf{k}^\mathsf{p}$.
\end{lem}
\begin{proof}
If the transposition is applied to any $\sI$ of $inf(\mathds{S}_\mathsf{k}^\mathsf{p})$, the transpose is a lower word which is consequently not balanced except for the last bit of $inf(\mathds{S}_\mathsf{k}^\mathsf{p})$, in this case, the transpose is $sup(\mathds{S}_\mathsf{k}^\mathsf{p})$. This result is consistent modulo rotation.
\end{proof}

If $\mathsf{k}$ and $\mathsf{p}$ are not relatively prime, $GCD(\mathsf{k},\mathsf{p})=x$. $\forall u\in \mathds{S}_\mathsf{k}^\mathsf{p}$, $u=v^x$. We define $\Delta=\Delta'$ such that $\Delta'$ is the unique location where $\tau(v,\Delta') \in \mathds{S}_\mathsf{k/x}^\mathsf{p/x}$. 

Lemma \ref{delayskp} shows that $\Delta$ is the last position of $inf(\mathds{S}_\mathsf{k}^\mathsf{p})$. Starting from this location, $\Delta$ can be found in every word of $\mathds{S}_\mathsf{k}^\mathsf{p}$ .

\begin{cor}
In $\rho^n(inf(\mathds{S}_\mathsf{k}^\mathsf{p}))$, $\Delta=\mathsf{p}+n\equiv n \mod \mathsf{p}$.
\end{cor}

We define the transposition function as the transposition applied on the bit $\Delta$ of a balanced binary word.

\begin{defn}[The transposition function on balanced binary words]\-\\
\label{def_taufunction}
We define the transposition function applied on balanced binary words as:
$\tau^n$: $\mathds{S}_\mathsf{k}^\mathsf{p} \to \mathds{S}_\mathsf{k}^\mathsf{p}$.
$\tau^0(u)=u$, $\tau(u)=\tau^1(u)=\tau(u,\Delta)$ where $\Delta$ is the same as in Lemma \ref{delayskp},  $\tau^n=\tau^{n-1} \circ \tau$, and $\tau^{-n}(u)=v$ if and only if $\tau^{n}(v)=u$.
If $\mathsf{k}$ and $\mathsf{p}$ are not relatively prime, $GCD(\mathsf{k},\mathsf{p})=x$. $\forall u\in \mathds{S}_\mathsf{k}^\mathsf{p}$, $u=v^x$. $\tau^n(u)=(\tau^n(v))^x$.
\end{defn}

\begin{exmp}
$\tau^1(1101010)=1011010$, $\tau^2(1010101)=0101101$, $\tau^\mathsf{p}(w)=w$, and $\tau(110110)=101101$.
\end{exmp}

\begin{lem}
The function $\tau^n$ is bijective.
\end{lem}
\begin{proof}
Since $\tau^1(\rho^n(inf(\mathds{S}_\mathsf{k}^\mathsf{p})))=\rho^n(sup(\mathds{S}_\mathsf{k}^\mathsf{p}))$, there is a one to one correspondence
between the elements and the images through the $\tau$ function.
\end{proof}

\subsection{Equivalence between rotation and transposition on balanced binary words}
\label{majorresult}

Theorem \ref{theooperation} presents our original result on balanced binary word. It states that for any given balanced binary word $u$, the transpose of $u$ is equivalent to the rotation of $u$ with a spin $-\alpha$. Let us first define $\alpha$.

\begin{defn}[The alpha coefficient]
\label{defalpha}
Let $\mathsf{k}$, $\mathsf{p}$ be two relatively prime integers, $0<\mathsf{k}<\mathsf{p}$.  
$\alpha$ is the inverse of $-\mathsf{k} \mod \mathsf{p}$. So we have $-\mathsf{k}*\alpha\equiv 1 \mod \mathsf{p}$ and $\alpha$ relatively prime with $\mathsf{p}$.
\end{defn}

\begin{thm}
\label{theooperation}
$\forall u \in \mathds{S}_\mathsf{k}^\mathsf{p}$,  $\rho^{-\alpha}(\tau(u)) = u$.
\end{thm}
\begin{proof}

We are going to prove that $u=u_1.0.1.u_2$ and $\rho^\alpha(u)=u_1.1.0.u_2$ ($u_1,u_2 \in \mathds{B}^*$). This means that $u$ is the transpose of $\rho^\alpha(u)$. So we compare $u$ and $\rho^\alpha(u)$ bit-wise for $i \in [\![1,\mathsf{p}]\!]$.  
$\rho^\alpha(u)(i)=u(i-\alpha)=\lfloor {(i-\alpha)*\mathsf{k} / \mathsf{p} }\rfloor - \lfloor{(i-1-\alpha)*\mathsf{k} / \mathsf{p} } \rfloor$.\\
$\alpha$ in $u(i-\alpha)$ is replaced by its value and the equation is simplified in:\\
$u(i-\alpha)=\lfloor {\frac{i*\mathsf{k}+1}{\mathsf{p}}}\rfloor - \lfloor{\frac{(i-1)*\mathsf{k}+1}{\mathsf{p}}} \rfloor$. Otherwise,
$u(i)=\lfloor \frac{i*\mathsf{k}}{\mathsf{p}}\rfloor - \lfloor \frac{(i-1)*\mathsf{k}}{\mathsf{p}}\rfloor$.
\-\\
For $i*\mathsf{k} \ne \mathsf{k}-1$ and $i*\mathsf{k} \ne \mathsf{p}-1$ modulo $\mathsf{p}$, $u(i-\alpha)=u(i)$ and\\
for $i*\mathsf{k}=\mathsf{p}-1$ modulo $\mathsf{p}$, $u(i-\alpha)=1$, $u(i)=0$, moreover,\\
$(i+1)*\mathsf{k}=\mathsf{p}-1+\mathsf{k}=\mathsf{k}+1$ modulo $\mathsf{p}$, and $u(i+1-\alpha)=0$, $u(i+1)=1$

\end{proof}

The proposed algorithm computes the schedules of the transitions from the schedules of its parent transitions. These schedules are equivalent by rotation because they are all balanced. Thanks to Theorem \ref{theooperation}, the rotation is used instead of transposition in the schedule computation formulas. This simplification lightens the formulas and allows correctness checking of the proposed algorithm.

\subsection{From word to schedule}
\label{subsection_relation}
The unitary forward rotation represents the effect of a latency on a transition schedule while the unitary forward transposition represents the effect of a delay. Figure \ref{PDTR} focuses on two transitions of a $4$-periodic MG with a period $7$. The schedules of $A$ and $B$ are binary words with length $7$ containing $4$ bits with the value $1$. In Figure \ref{PDTR}-a, the schedule of $B$ is the unitary rotation of the schedule of $A$ because no delay is affected to the place in-between. The arrows illustrate this rotation ($B(i+1)=A(i)$, $\forall i \mod 7$). In Figure \ref{PDTR}-b, two delays are affected to the place in-between. $B$ does not compute all the tokens generated by $A$ as soon as they are available any more. Two of them are delayed. The schedule of $B$ is the double transposition of the rotation of the schedule of $A$. The first arrow in diagonal illustrates the rotation, the two next, the transpositions. In Figure \ref{PDTR}-c, thanks to Theorem \ref{theooperation}, the succession of operations presented in Figure \ref{PDTR}-b is replaced by the equivalent rotation of value: $1-2*\alpha$ where $1$ is the original rotation, $-2*\alpha$ represents the two transpositions. For $(\mathsf{k},\mathsf{p})=(4,7)$, we have $\alpha=5$ (Definition \ref{defalpha}). So the spin of the rotation is $5$ ($1-2*\alpha \equiv 1-2*5 \equiv -9\equiv 5 \mod 7$).

\begin{figure}[hbpt]
 \begin{center}
   {\includegraphics[scale=0.42]{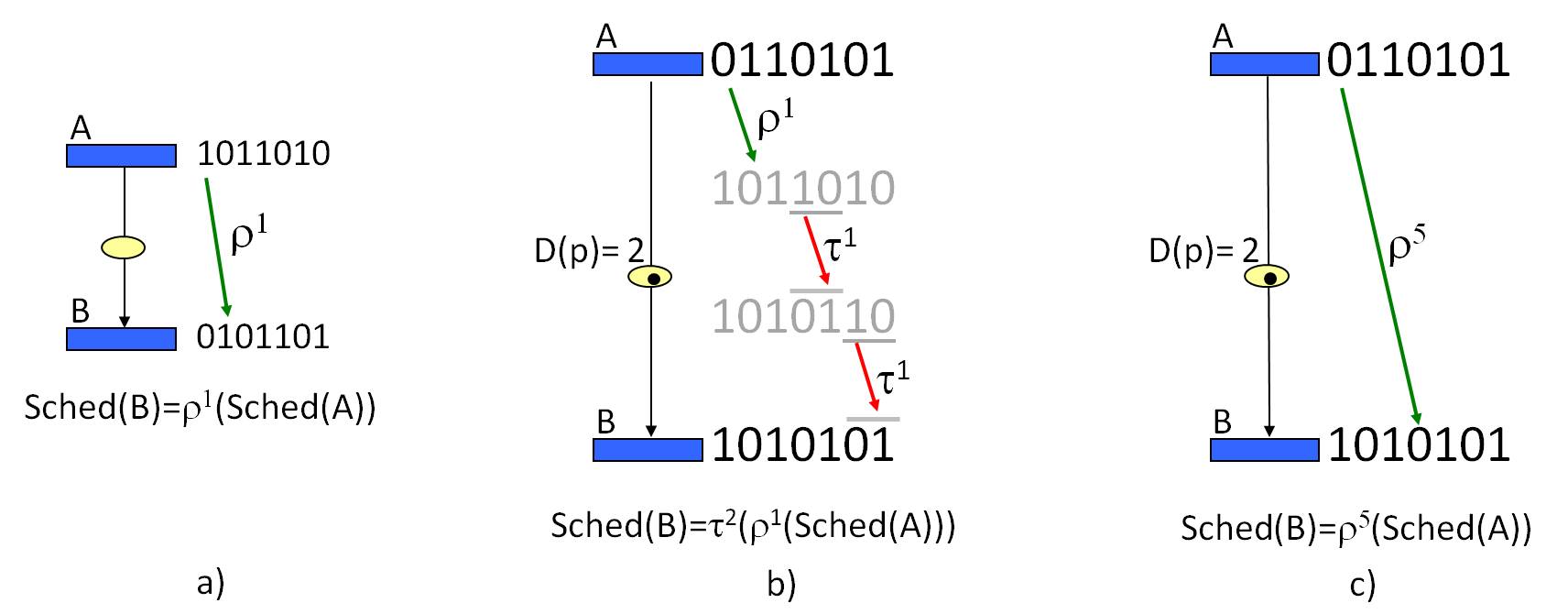}}
  \caption{a) The schedule of $B$ is the unitary rotation of the schedule of $A$. b) The schedule of $B$ is the double transposition of the rotation of the schedule of $A$ because the place in between is a 2-delays place. c) Thanks to Theorem \ref{theooperation}, the schedule of $B$ in b) is the rotation of spin $5$ of the schedule of $A$.} 
  \label{PDTR}
  \end{center}
\end{figure}

\section{Balanced scheduling of MG}
\label{section_bal_stat_sched}
This section details the proposed algorithm that computes an execution which is characterized by the following properties: i) the execution rate is maximal, ii) place sizes are minimal, and iii) after a guided initialization, the execution is ASAP. 

{\bf Input:} the proposed algorithm, presented in Algorithm 1, takes as input a live and strongly connected MG with a throughput inferior or equals to $1$. Section \ref{ssec_simpyconnected} discusses the application of the proposed algorithm on a simply connected MG. 

{\bf Output:} Algorithm 1 returns the computed execution along with the size of the places required for this execution. 

The following notation are used in Algorithm 1:
\begin{itemize}
\item $G$ is the MG in input and $M_0$ is its initial marking.
\item $D$ is the latest delays position (Definition \ref{def_latestdelaypos}).
\item $Exec_{initial}$ is the initial guided execution of $G$ from its initial marking to $M_{periodic}$.
\item $M_{periodic}$ is the marking of $G$ from which $Exec_{periodic}$ starts. 
\item $Exec_{periodic}$ is an balanced ASAP execution of $G$ from the marking $M_{periodic}$. 
\item $Sched(t)$ is the schedule of the transition $t$ in $Exec_{periodic}$.
\item The execution $Exec=Exec_{initial}.Exec_{periodic}$ is the output of the proposed algorithm.
\item $C_{Exec}$ gives place sizes according to $Exec$ (Definition \ref{def_capa}).
\end{itemize}

We consider that the preliminary step of the proposed algorithm is the $\mathds{N}$-equalization of the MG followed by the expansion of its latencies.
$\mathds{N}$-equalization is discussed in Section \ref{subsec_eqn} and expansion of latencies is discussed in Section \ref{ssec_latexp}.

\begin{algorithm}[htb]
\label{algo_main}
\caption{The proposed algorithm}
\begin{algorithmic}
\STATE \textbf{Input :} $G$ with its initial marking $M_0$.\\
\STATE \textbf{Output :} The execution $Exec$ and the place sizes $C_{Exec}$.\\
\STATE 1. $(\mathsf{k},\mathsf{p})\gets compute\_k\_p(G)$
\STATE 2. $D\gets compute\_D(G)$
\STATE 3. $Exec_{periodic}\gets compute\_Exec_{periodic}(G,D,\mathsf{k},\mathsf{p})$
\STATE 4. $M_{periodic}\gets compute\_M_{periodic}(G,Exec_{periodic},\mathsf{p})$
\STATE 5. $Exec_{initial}\gets compute\_Exec_{initial}(G,M_0,M_{periodic})$
\STATE 6. $Exec\gets Exec_{initial}.Exec_{periodic}$
\STATE 7. $C_{Exec}\gets compute\_C_{Exec}(G,D,\mathsf{k},\mathsf{p})$.
\RETURN $(Exec,C_{Exec})$
\end{algorithmic}
\end{algorithm}


\subsection{Algorithm details}

\subsubsection{Step 1: compute $\mathsf{k}$ and $\mathsf{p}$}
The formula is given in \cite{QuadratCohenBaccelli92}. $\mathsf{k}=GCD(M_0(c))$ and $\mathsf{p}=GCD(L(c))$, for all cycle $c$ of the CSCCs. Step 1 requires the enumeration of all the elementary cycles. This enumeration has an exponential complexity with respect to the number of transitions. It binds the overall complexity of the proposed algorithm.

\subsubsection{Step 2: compute the latest delays position $D$}
\label{sssection_delay}
$D$ has to be the latest delays position (Definition \ref{def_latestdelaypos}) in order to build the ASAP execution $Exec_{periodic}$. 
Theorem \ref{thm_ldp_ex} shows that the latest delays position can be deduced from any ASAP execution of $G$. Thus, Step 2 computes $D$ from the ASAP execution of $G$. Step 2 has a polynomial complexity according to the number of transitions. Algorithm 2 details Step 2.

\begin{algorithm}[htb]
\label{algo_step_2}
\caption{$compute\_D$}
\begin{algorithmic}
\STATE \textbf{Input :} $G$.
\STATE \textbf{Output :} $D$.
\STATE Run the ASAP execution of $G$.
\FORALL{$p\in P$}
\STATE $D(p)=\Sigma_{i=1}^{\mathsf{p}}Delay(p,j_0+i)$ where $j_0$ is the length of the initial part.
\ENDFOR
\WHILE{$D$ is not the latest delay position}
\FORALL{$t\in T$}
\STATE $forwarded\_delay=min(D(p)\mid \forall p\in{}^\bullet t)$
\FORALL{$p\in {}^\bullet t$}
\STATE $D(p)-=forwarded\_delay$
\ENDFOR
\FORALL{$p\in t^\bullet$}
\STATE $D(p)+=forwarded\_delay$
\ENDFOR
\ENDFOR
\ENDWHILE 
\RETURN $D$
\end{algorithmic}
\end{algorithm}

Figure \ref{figure_delay} presents $D$ on the running example. The right-most cycle, $c_1$, is critical, it does not contain any delay. The left-most cycle, $c_2$, is not. The difference of firing over a period is $|c_1|*|c_2|_1-|c_1|_1*|c_2|=7*2-4*3=2$. The places of $c_2$ that do not belong to $c_1$ should share $2$ delays. The left-most and top-most place contains all these delays because in the latest delays position, the delays have to occur as late as possible.

\begin{figure}[hbpt]
 \begin{center}
   {\includegraphics[scale=0.6]{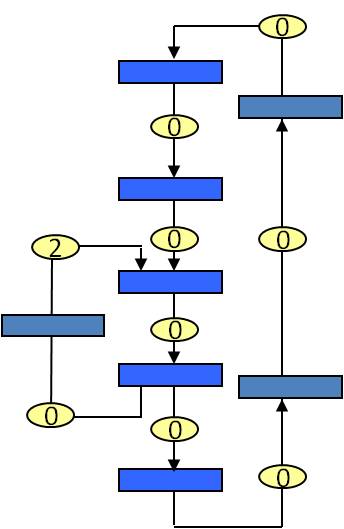}}
  \caption{$D$: The amount of delay is written within the place.} 
  \label{figure_delay}
  \end{center}
\end{figure}

\subsubsection{Steps 3: compute $Exec_{periodic}$}
\label{sssection_sched}
Step 3 affects a schedule to every transition with respect to $D$. Algorithm 3 details Step 3. It has a linear complexity according to the number of transitions.

\begin{algorithm}[htb]
\label{algo_step_3}
\caption{$compute\_Exec_{periodic}$}
\begin{algorithmic}
\STATE \textbf{Input :} $G$, $D$, $\mathsf{k}$, and $\mathsf{p}$.
\STATE \textbf{Output :} $Exec_{periodic}$.
\STATE Let $t\in T$, $Sched(t)\gets get\_a\_word\_in(\mathds{S}_\mathsf{k/r}^\mathsf{p/r})$ \COMMENT{with $r=GCD(\mathsf{k},\mathsf{p})$.}
\STATE $current\_transition\gets t$
\WHILE{$\exists t'\in T$ such that $Sched(t')$ is not defined}
\FORALL{$t'\in \{(current\_transition^\bullet)^\bullet\}$}
\STATE $Sched(t')\gets\rho^{1-D({}^\bullet t')*\alpha}(Sched(current\_transition))$
\ENDFOR
\STATE $current\_transition\gets t'$
\ENDWHILE 
\RETURN $Exec_{periodic}$
\end{algorithmic}
\end{algorithm}

In Figure \ref{figure_schedule}, Step 3 generates a balanced binary word $1101010\in \mathds{S}_4^7$ because the MG is $4$ periodic with a 
period $7$. Step 3 affects this word to a transition and it computes the schedule of the other transitions using the rotation.
The schedule of the 2-inputs transition ($1010101$) can be found from its right predecessor $\rho^1(0101011)$ or from its left predecessor $\rho^5(0110101)$.
The spin of this last rotation is $5\equiv 1-2*\alpha\mod \mathsf{p}$. The place in-between the transitions contains $2$ delays. Since $\alpha=5$, $1-2*\alpha=1-2*5=-9\equiv 5 \mod 7$ .

\begin{figure}[hbpt]
 \begin{center}
   {\includegraphics[scale=0.6]{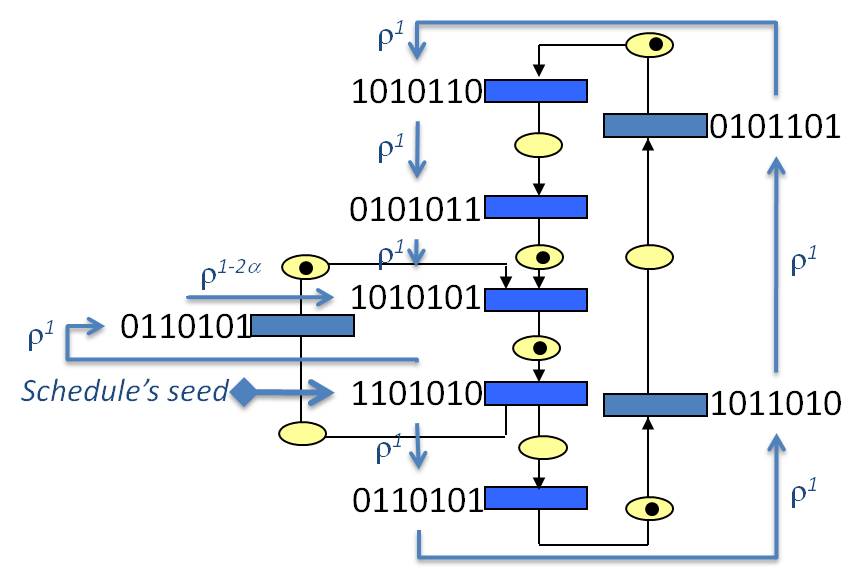}}
  \caption{From the schedule's seed, Step 3 generates all other schedules through rotation. The schedule of the 2-inputs
  				transition can be found from both its predecessor.} 
  \label{figure_schedule}
  \end{center}
\end{figure}

The consistency of this method is guaranteed because the number of delay for each cycle is conformed to Theorem \ref{th_delay_cycle}. The lemma \ref{lem_schedaffectation} formalizes this result.

\begin{lem}[Creation of $Exec_{periodic}$]
\label{lem_schedaffectation}
Step 3 is consistent.
\end{lem}
\begin{proof}
Let $u\in \mathds{S}_\mathsf{k}^\mathsf{p}$ be a balanced binary word. The number of delays occurring on a cycle $c$ during a period of execution is $n=M_0(c)*\mathsf{p}-L(c)*\mathsf{k}$. The latency on this same cycle is $L(c)$. 

If we impose the schedule of a transition $t$ on $c$ to $Sched(t)=u$ and we propagate this schedule to the successors according to Step 3, then $t$ will be ultimatly reached again.
The updated schedule of the $t$ will be $\rho^{L(c)-n*\alpha} u$. We know from Definition \ref{defalpha} that $\alpha*\mathsf{k}\equiv -1\mod \mathsf{p}$ so if we focus on the quantity $L(c)-n*\alpha$:\\
$L(c)-n*\alpha= L(c)-\alpha*(M_0(c)*\mathsf{p}-L(c)*\mathsf{k})    \equiv    L(c)-\alpha*M_0(c)*\mathsf{p}+\alpha*L(c)*\mathsf{k}$\\
$\equiv L(c)-\alpha*M_0(c)*\mathsf{p}-L(c)\mod \mathsf{p}    \equiv    -\alpha*M_0(c)*\mathsf{p}\mod \mathsf{p}\equiv 0\mod \mathsf{p}$, it is equivalent to $0$ modulo $\mathsf{p}$.

Consequently, the schedule of $t$ remains the same, the method is consistent.
\end{proof}

\subsubsection{Step 4: compute $M_{periodic}$}
\label{sssection_marking}
Step 4 deduces $M_{periodic}$ from $Exec_{periodic}$. 
$M_{periodic}$ is not only the marking from which $Exec_{periodic}$ runs but also the marking generated by
$Exec_{periodic}$ after a period of execution. 
Consequently, the last step of a period reaches $M_{periodic}$.
The last bit of $Sched({}^\bullet p)$ represents the activity of ${}^\bullet p$ at the last instant of the period. If it has been active, it has produced a token in $p$. 
Algorithm 4 details Step 4. It has a linear complexity according to the number of places.

\begin{algorithm}[htb]
\label{algo_step_4}
\caption{$compute\_M_{periodic}$}
\begin{algorithmic}
\STATE \textbf{Input :} $G$, $Exec_{periodic}$, and $\mathsf{p}$.
\STATE \textbf{Output :} $M_{periodic}$.
\STATE $p\in P$, $Sched({}^\bullet p)=u^\omega$ and $Sched(p^\bullet)=v^\omega$
\FORALL{$p\in P$}
\STATE $M_{periodic}(p)\gets u(\mathsf{p})+[\rho(u)<v]$\\ \COMMENT{$[\rho(u)<v]=1$ if $\rho(u)<v$ and $0$ otherwise.}
\ENDFOR
\RETURN $M_{periodic}$
\end{algorithmic}
\end{algorithm}

$[\rho(u)<v]=1$ means that one token is being delayed in the place at the current instant. 
$[\rho(u)<v]$ is always equal to $0$ when $D(p)=0$ because $v=\rho(u)$. 
When $D(p)>0$, $v$ is the transpose of $\rho(u)$. In the usual case, $\rho(u)>v$ because transposition shifts $\sI$s to the right. But when the transposition occurs on the last bit of the word, the transpose gets a bit on its first position and becomes higher than the original word. Thus, if a transposition occurs on the last bit, it means that a token is currently delayed in the place. Lemma \ref{lemma_delaydiff} formalizes this intuition. 

\begin{lem}[Presence of tokens in delayed places]
\label{lemma_delaydiff}
Let $p$ be a place of $G$ such that $D(p)=n>0$.
Let $u=Sched({}^\bullet p)$ and $v=Sched(p^\bullet)$.
If $\rho(u)<v$, $p$ is delaying a token in the marking $M_{periodic}$.
\end{lem}
\begin{proof}
$v=\rho^{1-n*\alpha}(u)=\tau^n(\rho(u))$. By definition, the transpose of a word is lower than the original word except when the last bit is transposed. In this last case, the transpose is higher that the original word. If, $v>\rho(u)$ (but $v=\tau^n(\rho(u))$), at least one of the transpositions occurs on the last bit. The interpretation of this statement is that the firing of $p^\bullet$ was supposed to occur at the last instant of the period but has been delayed to the next one. The token related to this execution is currently in $p$.
\end{proof}

Figure \ref{figure_marking} illustrates Step 4. The last bit of the schedule of a transition determines whether a token is present in its output place(s). The place with delays contains a regular token because the schedule of the predecessor finishes by $\sI$ but it does not contain an extra token because $1010101<\rho(0110101)=1011010$.

\begin{figure}[hbpt]
 \begin{center}
   {\includegraphics[scale=0.6]{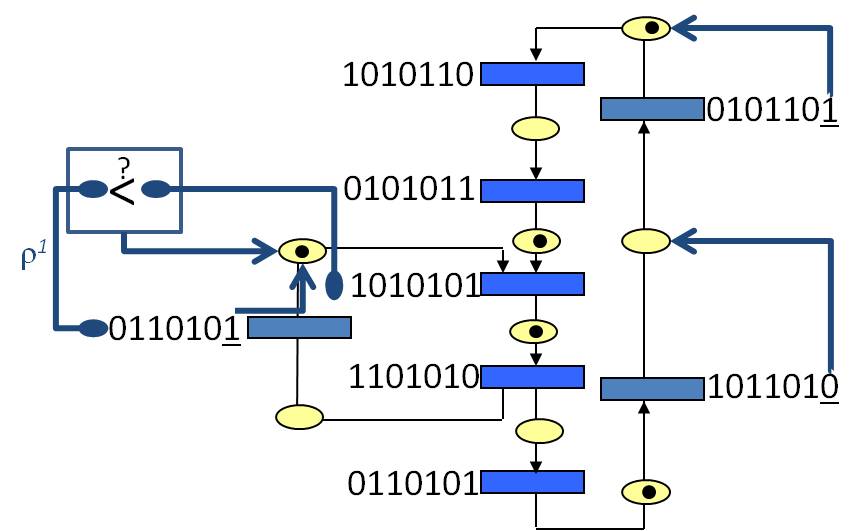}}
  \caption{The step 4 generates $M_{periodic}$ from $Exec_{periodic}$. The presence of an additional token in the delayed place is
  found using the function $[v>\rho(u)]$.} 
  \label{figure_marking}
  \end{center}
\end{figure}

The correctness of Step 4 is presented in Section \ref{sec_corr_prop}. First, Lemma \ref{lem_reachmark} proves that the marking $M_{periodic}$ is reachable from $M_0$. Then, Theorem \ref{lem_tokengame} shows that the ASAP execution of $G$ from $M_{periodic}$ is $Exec_{periodic}$.

\subsubsection{Step 5: compute $Exec_{initial}$}
\label{sssection_init}
Algorithm 5 computes $Exec_{initial}$ based on integer linear programming solving. 
The optimization criterion is the minimization of the number of firing because one cannot express linearly the minimization of the number of steps required to run $Exec_{initial}$. The mapping $F_{init}$ associates to each transition the number of firing required to reach $M_{periodic}$. The function $build\_execution$ builds $Exec_{initial}$ by simulating an ASAP execution of $G$ where each transition $t$ cannot be fired more than $F_{init}(t)$. The complexity of Step 5 depends upon the algorithm used to solve the linear system of inequation. Lemma \ref{lem_init_correct} shows the correctness of Step 5.

\begin{algorithm}[htb]
\label{algo_step_5}
\caption{$compute\_Exec_{initial}$}
\begin{algorithmic}
\STATE \textbf{Input :} $G$, $M_0$, and $M_{periodic}$.
\STATE \textbf{Output :} $Exec_{initial}$.
\STATE $Cst=\O$\COMMENT{$Cst$ is the set of linear constraints}
\FORALL{$t\in T$}
\STATE $Cst+=\{F_{init}(t)\ge 0\}$
\FORALL{$p\in t^\bullet$}
\STATE $Cst+=\{F_{init}({}^\bullet p)=F_{init}(p^\bullet)+M_{periodic}(p)-M_0(p)\}$
\ENDFOR
\ENDFOR
\STATE $F_{init}\gets lp\_solve(Cst,Min(\Sigma_{\forall t\in T} F_{init}(t)))$
\STATE $Exec_{initial}\gets build\_execution(F_{init})$
\RETURN $Exec_{initial}$
\end{algorithmic}
\end{algorithm}

In Figure \ref{figure_initialization2}, $M_0$ is on the left. The $2$-bits-length schedules attached to each 
transition is $Exec_{initial}$ leading to $M_{periodic}$ on the right. 

\begin{lem}[Correctness of Step 5]
\label{lem_init_correct}
Algorithm 5 computes a valid execution $Exec_{initial}$ reaching $M_{periodic}$.
\end{lem}
\begin{proof}
Let us call $M_1$ the marking at the end of $Exec_{initial}$. $\forall p\in P$, $M_1(p)=M_0(p)-F_{init}(p^\bullet)+F_{init}({}^\bullet p)=M_0(p)-F_{init}(p^\bullet)+F_{init}(p^\bullet)+M_{periodic}(p)-M_0(p)=M_{periodic}(p)$.
\end{proof}

According to \cite{deselEsparza}, the maximum number of firings between two markings ($M_0$ and $M_{periodic}$ in our case) is
in $O(n^3)$ where $n$ is the number of transitions in the MG. We assume that the length of $Exec_{initial}$ is convenient because:
i) the bound $O(n^3)$ is given in terms of number of firings. $Exec_{initial}$ allows parallel firing of transitions.
ii) the periodic execution $Exec_{periodic}$ covers a set of $\mathsf{p}$ markings. The initial part can reach any of these marking. 
So the problem is equivalent to: reaching the closest marking of $Exec_{periodic}$ instead of only $M_{periodic}$.
iii) the cases where the upper bound is reached are extreme cases where all tokens have to shift to another place
far from the initial one or because the shift of one token implies the shift of all others. In $M_{periodic}$, the tokens are ``spread equally" in the MG. $M_{periodic}$ might be the easiest reachable marking.

\subsubsection{Step 6: compute $Exec$}
\label{sssection_conclusion}
$Exec$ is composed of $Exec_{initial}$ followed by $Exec_{periodic}$. After the guided initialization, the execution is ASAP and repetitive. 
In Figure \ref{figure_final}, the MG is in its initial marking. The execution, $Exec$, is represented by the ultimately $\mathsf{k}$-periodic schedules attached to each transition.

\subsubsection{Step 7: compute $C_{Exec}$}
If a place does not contain delay, every token reaching the place leaves it at the next instant. As long as a place contains at most one token in $M_0$, its size is $1$. Lemma \ref{lem_size_of_place} demonstrates that if a place contains delays, tokens are never delayed more that one consecutive instant because the MG is $\mathds{N}$-equalized and the schedules are balanced. In consequence a place cannot accumulate more than two tokens. 

\begin{lem}[Delayed place size is bounded by $2$]
\label{lem_size_of_place}
According to $Exec$, place size where delays occur is bounded by 2.
\end{lem}
\begin{proof}
First, $G$ is $\mathds{N}$-equalized, so the number of delay per place is bounded by $\mathsf{k}$.
Secondly, since the execution is balanced, a token can be delayed only once in a row.
Lastly, since the execution is $\mathsf{k}$-periodic, there is (at most) $\mathsf{k}$ different tokens to delay.
These conditions guarantee that a token cannot stay more than 2 instants in the place. Consequently, no accumulation of more than 2 tokens can occur.
\end{proof}

Even for delayed places, a size of two is required only if a token is delayed while another reaches the place. Theorem \ref{lem_size_1} shows that a delayed place has a size of one when $D(p)<\mathsf{p}-\mathsf{k}$ because delays occur first on the $\sI$ which are followed by a $\sO$. In Figure \ref{figure_delay}, all the places have a size of $1$. In the delayed place $p$, $D(p)=2<7-4$. 

\begin{thm}[Exact delayed place size]
\label{lem_size_1}
Let $p$ be a place, 
\begin{equation*}
C_{Exec}(p)=1 \Leftrightarrow D(p)\le \mathsf{p}-\mathsf{k}
\end{equation*}
\end{thm}
\begin{proof}
First, if a place $p$ with $D(p)=n$ has a size one, every other place $p'$ with $D(p')\le n$ also has a size one.
If a place $p$ with $D(p)=m$ has a size two, every other place $p'$ with $D(p')\ge m$ also has a size two.
This property is guaranteed by the Lemma \ref{delayskp}. In two different delayed places, the delayed tokens
are the same modulo rotation. So the problem of calibrating the size of a place only depends upon the amount of delays 
in that place and not at all about the location of these delays. 

Let $u=Sched({}^\bullet p)$ and $v=Sched(p^\bullet)$. A place $p$ requires a size two when a token is used after the
next one has reached the place.
Formally, there exists $n$ such that $[v]_n>[u]_n$ (where $[u]_n$ is the position of the $n^{th}$ $\sI$ in $u$). 
$v$ says when the current token is used, $u$ says when a new token reaches $p$.

Let us assume that $D(p)= \mathsf{p}-\mathsf{k}$.
We have $v=\rho^{1-(p-k)*\alpha} u = \rho^{1-p*\alpha+k*\alpha}u=u$ so $[u]_n>[u]_n$ never holds.

Let us assume that $D(p)= \mathsf{p}-\mathsf{k}+1$.
We have $v=\rho^{1-(p-k+1)*\alpha} u = \rho^{1-p*\alpha+k*\alpha-\alpha}u=\tau(u)$ so $[\tau(u)]_n>[u]_n$ holds when $n$ 
is the index of the delayed token.
\end{proof}

The following theorem proves that the proposed algorithm computes an execution which has a minimal size of the places as claimed earlier.

\begin{thm}[Minimal size of the places]
$C_{Exec}$ gives the minimal size of places.
\end{thm}

\begin{proof}
When $D(p)\le \mathsf{p}-\mathsf{k}$, $C_{Exec}(p)=1$ so it is minimal.

Let us now assume an ASAP execution $Exec'$ from the marking $M'$ reachable from $M_0$.
Let assume a place $p'$ such that $D(p')= \mathsf{p}-\mathsf{k}+1$. 
At most $\mathsf{p}-\mathsf{k}$ tokens within a period can be delayed while no token follows.
It remains at least $1$ token that has to be delayed but that is followed by another token.
In this last configuration, $p$ contains two tokens and thus the size of $p$ is at least 2.
Consequently, $C_{Exec}(p)$ is also minimal when $D(p')> \mathsf{p}-\mathsf{k}$.
\end{proof}

\subsection{Correctness of the step 4}
\label{sec_corr_prop}

Let us first prove the reachability of $M_{periodic}$ from $M_0$ then we prove that the ASAP execution from $M_{periodic}$ is $Exec_{periodic}$. 

\subsubsection{Reachability of $M_{periodic}$ from $M_0$}

\begin{lem}[Reachability of $M_{periodic}$ from $M_0$]
\label{lem_reachmark}
$M_{periodic}$, as computed in the step 4, is reachable from $M_0$.
\end{lem}
\begin{proof}
According to \cite{deselEsparza}, both markings are mutually reachable if and only if for each cycle of the MG, the two markings have the same number of tokens. Now, let us prove that $M_{periodic}$ and $M_0$ respect this condition.

First, Lemma \ref{lem_reach_simple} considers that all the delays of a cycle are assembled in the same place and proves that the condition holds. Lemma \ref{lem_reach_simple} requires the Lemma \ref{lem_suff}. Then Lemma \ref{lem_reach_general} generalizes Lemma \ref{lem_reach_simple} to any allocation of delays in a cycle.
\end{proof}

If all the delays are assembled in the same place $p$, $M_{periodic}(c)$ is equals to the number of $\sI$s in the suffix of length $L(c)$ of $Sched(p^\bullet)$ because the schedules are, in such a case, elementary rotations of the previous ones and the bit of index $\mathsf{p}$ says whether a token is there in the output place. We have seen in Section \ref{subsection_defBBW} that the number of $\sI$s in a factor of a balanced binary word of length $L(c)$ is either $\lfloor {L(c)*|u|_1 /|u|}\rfloor$ or $\lceil {L(c)*|u|_1 /|u|}\rceil$. Lemma \ref{lem_suff} proves that if the suffix of length $L(c)$ has $\lfloor {L(c)*|u|_1 /|u|}\rfloor$ $\sI$s, $p$ is currently delaying a token. Otherwise, $p$ is not. Consequently, the number of tokens in $c$ is always $\lceil {L(c)*|u|_1 /|u|}\rceil$. Lemma \ref{lem_reach_simple} concludes that if the MG is equalized, $M_0(c)=\lceil {L(c)*|u|_1 /|u|}\rceil$ also.

\begin{lem}[Suffixes and lexicographic order in $\mathds{S}_\mathsf{k}^\mathsf{p}$]
\label{lem_suff}
Let $u\in \mathds{S}_\mathsf{k}^\mathsf{p}$ and $j, l\in \mathds{N}$ such that $0<j\le l$ and $\mathsf{k}>j*\mathsf{p}-\mathsf{k}*l\ge 0$. We note $n=j*\mathsf{p}-\mathsf{k}*l$.

There exists $n$ balanced binary words $v\in O(u)$ such that $|suffix(v,l)|_1=\lfloor l*\mathsf{k}/\mathsf{p}\rfloor$ ($suffix(v,l)$ is the suffix of $v$ of length $l$). Moreover, these $n$ words are the highest according to the lexicographic order.
\end{lem}
\begin{proof}
Consider the word $u^l$. By definition $slope(u^l)=slope(u)=\mathsf{k}/\mathsf{p}$. $u^l$ can be sliced in $\mathsf{p}$ factors of length $l$. Each factor is different from the others and matches with a suffix of length $l$ of $v\in O(u)$. If the number of factors containing $\lfloor l*\mathsf{k}/\mathsf{p}\rfloor$ $\sI$s is different from $n$, $slope(u^l)$ cannot be $\mathsf{k}/\mathsf{p}$.

Moreover, if $|suffix(v,l)|_1=\lfloor l*\mathsf{k}/\mathsf{p}\rfloor$, $|prefix(v,\mathsf{p}-l)|_1=\mathsf{k}-\lfloor l*\mathsf{k}/\mathsf{p}\rfloor$. So if 
$|suffix(v,l)|_1=\lceil l*\mathsf{k}/\mathsf{p}\rceil$, $|prefix(v,\mathsf{p}-l)|_1=\mathsf{k}-\lceil l*\mathsf{k}/\mathsf{p}\rceil$. A word with more $\sI$s in its prefix is higher than another with less $\sI$s according to the lexicographic order.
\end{proof}

\begin{lem}[Reachability of $M_{periodic}$ from $M_0$ in the simple case]
\label{lem_reach_simple}
Let $c$ be a cycle of $G$ such that all the delays occurring in $c$ are assembled in the place $p$.
We have $M_{periodic}(c)=M_0(c)$.
\end{lem}
\begin{proof}
Let us call $u$ the schedule of $p^\bullet$.
The number of token in $c$ is $M_{periodic}(c)=\Sigma_{i=0}^{L(c)-1} u(\mathsf{p}-i) + [u>\rho^{D(p)*\alpha}u]$.

$\Sigma_{i=0}^{L(c)-1} u(\mathsf{p}-i)=|suffix(u,L(c)|_1$. Since $u$ is balanced, $\lfloor L(c)*\mathsf{k}/\mathsf{p}\rfloor\le |suffix(u,L(c)|_1 \le \lceil L(c)*\mathsf{k}/\mathsf{p}\rceil$.

{\bf Case 1}: if $|suffix(u,L(c))|_1=\lfloor L(c)*\mathsf{k}/\mathsf{p}\rfloor$, $u$ is one of the $D(p)$ highest word of $O(u)$ (Lemma \ref{lem_suff}). Consequently, $u>\rho^{D(p)*\alpha}u$ because a rotation of $\alpha$ increases the value of the word according to the lexicographic order but if the highest is reached, another rotation of $\alpha$ gives the lowest. So $[u>\rho^{D(p)*\alpha}u]=1$ and $M_{periodic}(c)=\lfloor L(c)*\mathsf{k}/\mathsf{p}\rfloor+1=\lceil L(c)*\mathsf{k}/\mathsf{p}\rceil$ (In the case $D(p)\ne 0$, $\mathsf{p}$ does not divide $\mathsf{k}*L(c)$).

{\bf Case 2}: if $|suffix(u,L(c)|_1=\lceil L(c)*\mathsf{k}/\mathsf{p}\rceil$, $u$ is not one of the $D(p)$ highest word of $O(u)$ (Lemma \ref{lem_suff}). Consequently, $[u>\rho^{D(p)*\alpha}u]=0$, and $M_{periodic}(c)=\lceil L(c)*\mathsf{k}/\mathsf{p}\rceil$ also.

{\bf Conclusion}: since $G$ is $\mathds{N}$-equalized, $M_0(c)/L(c)\ge \mathsf{k}/\mathsf{p} >M_0(c)/(L(c)+1)$. So $(\mathsf{k}*L(c)+\mathsf{k})/\mathsf{p} > M_0(c) \ge \mathsf{k}*l/\mathsf{p}$. By definition of the $\mathds{N}$-equalization, the solution always exists and is unique: $\lceil L(c)*\mathsf{k}/\mathsf{p}\rceil$.
\end{proof}

In Lemma \ref{lem_reach_simple}, a delay can occurs only in one place but in Lemma \ref{lem_reach_general}, every place can contain delays and they might be delaying a token in $M_{periodic}$. In this Lemma, we give the formula to compute $M_{periodic}$ from a place $p_0$ that we are going to consider as the first place of the cycle, then we prove that if a delay is shifted to the last place of the cycle, the number of tokens in the cycle will be the same. Thanks to this result, we can shift all the delays into the last place and conclude that the number of tokens found in Lemma \ref{lem_reach_simple} is applicable to the general case. The inertia of the shift operation on the number of tokens is proven by considering the last places of the cycle such that the first and the last of this sequence of places contain delays but none of the other in-between does. In such a case, the effect of the shift operation on the formula to compute $M_{periodic}$ can be analyzed locally.

\begin{lem}[Reachability of $M_{periodic}$ from $M_0$ in the general case]
\label{lem_reach_general}
For all cycle $c$, $M_{periodic}(c)=M_0(c)$.
\end{lem}
\begin{proof}
Let $c$ be a cycle of $G$. The places of $c$ are $\{p_0,p_1,...,p_{L(c)-1}\}$.
We note $u$ the schedule of the transition ${}^\bullet p_0$.

$M_{periodic}(c)=\Sigma_{i=0}^{L(c)-1} \Big( u(p-(i-(D(p_0)+...+D(p_{i}))*\alpha)) +$\\ $[\rho^{i+1-(D(p_0)+...+D(p_{i+1}))*\alpha}u>\rho^{i+1-(D(p_0)+...+D(p_{i}))*\alpha}u] \Big)$.

Let $i_0$ be such that $\forall i\in ]i_0,L(c)-1]$, $D(p_{i})=0$ and let us focus on the few last terms of this sum such that $i_0<i\le L(c)-1$ (In the worst case, $i_0=L(c)-2$ and only the last term of the sum is there). The following equality is going to be proved for these terms only:\\
$[\rho^{i_0-(D(p_0)+...+D(p_{i_0}))*\alpha}u>\rho^{i_0-(D(p_0)+...+D(p_{i_0-1}))*\alpha}u]$ (A)\\
$+\Sigma_{i=i_0}^{L(c)-1} u(p-(i-(D(p_0)+...+D(p_{i}))*\alpha))$ (B)\\
$+[u>\rho^{D(p_{L(c)-1})*\alpha}u]$ (C)\\
$=$\\
$[\rho^{i_0-(D(p_0)+...+D(p_{i_0}-1))*\alpha}u>\rho^{i_0-(D(p_0)+...+D(p_{i_0-1})-1)*\alpha}u]$ (A')\\
$+\Sigma_{i=i_0}^{L(c)-1} u(p-(i-(D(p_0)+...+D(p_{i})-1)*\alpha))$ (B')\\
$+[u>\rho^{(D(p_{L(c)-1})+1)*\alpha}u]$ (C').

There is only three cases to study to prove this property:
\begin{itemize}
	\item When (A) is equals to $1$  but (A') is equals to $0$, then the first term of (B) is equals to $0$ and the first term of (B') is equals to $1$.
				If the first place delays a token (A)=1 but not any more after the shift (A')=0, the token has been computed instead of being delayed and then 				 it appears in the next place (B')=1. All the other term of the sum are the same.
	\item When (C) is equals to $0$ but (C') is equals to $1$, the last term of (B) is equals to $1$ and the last term of (B') is equals to $0$.
				If the last place does not delay any token (C)=0 but does after the shift (C')=1, this token was in the last but one place (B)=1 and is now
				in the last one (B')=0. All the other term of the sum are the same.
	\item In every other possible cases, (A) equals (A'), (B) equals (B'), (C) equals (C').
\end{itemize}

Thanks to this property, we know that the number of tokens in $c$ is the same wherever are the delays in the cycle. So the result found in lemma \ref{lem_reach_simple} is applicable to the general case.
\end{proof}

\subsubsection{Validity of $Exec_{periodic}$ from $M_{periodic}$}

\begin{lem}[A step of execution from $M_{periodic}$]
\label{lem_one_step}
Let $M_1$ be the marking resulting from a step of ASAP execution from $M_{periodic}$,
$M_1'$ is the marking resulting from a step of $Exec_{periodic}$ from $M_{periodic}$.

Then, $M_1=M_1'$
\end{lem}
\begin{proof}
In an ASAP execution, a transition $t$ executes if and only if all the incoming places contains a token.
In $M_{periodic}$, the place ${}^\bullet t$ contains a token if and only if $Sched({}^\bullet {}^\bullet t)(\mathsf{p})=1$ or $[\rho^1(Sched({}^\bullet {}^\bullet t))<Sched(t)]$.
In the first step of $Exec_{periodic}$, a transition $t$ executes if and only if $Sched(t)(1)=1 \Leftrightarrow \rho^{-1}(Sched(t))(\mathsf{p})=1 \Leftrightarrow Sched({}^\bullet{}^\bullet t)(\mathsf{p})=1$ or that $[\rho^1(Sched({}^\bullet{}^\bullet t))<Sched(t)]$. The condition of execution are the same. If the same transitions are fired according to an ASAP execution or $Exec_{periodic}$, then the resulting markings are the same.
\end{proof}

\begin{thm}[Validity of $Exec_{periodic}$]
\label{lem_tokengame}
The ASAP execution of $G$ from the marking $M_{periodic}$ is $Exec_{periodic}$.
\end{thm}
\begin{proof}
Step 3 is based on the affectation of a schedule by a random balanced binary word from $\mathds{S}_\mathsf{k}^\mathsf{p}$. 
The lemmas \ref{lem_schedaffectation}, \ref{lem_reachmark} and Lemma \ref{lem_one_step} also hold for any other balanced binary word from $\mathds{S}_\mathsf{k}^\mathsf{p}$. Since all the words of $\mathds{S}_\mathsf{k}^\mathsf{p}$ are equivalent by rotation, Step 4 gives all the successive markings of $Exec_{periodic}$ when the Step 3 is initiated with, successively, all the words of $\mathds{S}_\mathsf{k}^\mathsf{p}$. For each of these marking, Lemma \ref{lem_one_step} proves that the next marking is reachable through ASAP execution. Consequently, from $M_{periodic}$, and after $\mathsf{p}$ steps of execution, $Exec_{periodic}$ reaches $M_{periodic}$.
\end{proof}

\subsection{Extension to the simply connected case}
\label{ssec_simpyconnected}
As we have seen in proposition \ref{pro_boundednotasap}, one cannot guaranty that an ASAP and bounded
execution exists for a given simply connected MG. Since a System-on-Chip cannot be designed with unbounded memories, the extension of the proposed algorithm to simply connected case preserves the bounded property at the expense of the ASAP property. The maximum execution rate is still preserved but the minimality of the size of places is altered.

A simply connected MG can be transformed into a strongly connected one by adding feedback paths. Thus, the proposed algorithm can be applied. To do so, we add to the MG some feedback paths which bind all the components together. The functional behavior of the system will be preserved but its scheduling will be over-constrained by the added feedback paths i.e. adding different feedback paths imply a different execution computed by the proposed algorithm. These feedback paths act as synchronization barriers.

There is different algorithmic solution to realize the transformation; however, the added feedback paths should not create a cycle with a throughput inferior to the critical one in the original MG. Otherwise, the maximal execution rate will not be achieved. It is easy to prove that the marking and the latency of the added feedback paths can always be adjusted so that the created cycles have a non-critical throughput.

The minimality of the size of the places is guaranteed for the original SCCs, but the size of the places on the original DAC depends upon the added feedback paths. One may find another set of feedback paths such that the size of places on the original DAC is less. We have not yet studied this optimization.

\paragraph{Open MG}
If a simply connected MG is open, one can consider that the system has global input(s) and output(s). In order to schedule the MG, it is transformed in a strongly connected one.
Consequently, the MG becomes closed.
The run of the proposed algorithm shall return a schedule for every source and sink. The schedule of 
a sink says when the system produces an output token and the schedule of a source says when the system
consumes an input token. Thus, the concerned input token has to be present when required. In \cite{millo08}, we state that the execution rates of the feeder and eater have to be the same in order to calibrate the capacity of the ``interconnection" place with a finite value and thus ensure on-demand token availability. In \cite{CMPP08}, the authors study thoroughly the sizing of buffer between clocked systems.

\paragraph{The AES example}
Figure \ref{figure_aes} presents an implementation of the AES encryption standard.
The MG has been represented using K-Passa (K-Periodic Asap Static Schedule Analyser) \cite{kpassa}.
K-Passa implements the proposed algorithm but also the $\mathds{N}$-equalization.
The circles represent the transitions of the system. The arrows represent the sequences ($arc \to place \to arc$) in-between two transitions. The two left most transitions called {\em key} and {\em word} are sources (the local loop has been added for simulation purpose). The central transition called {\em output word} is a sink. The schedule attached to each transition is the one computed by the proposed algorithm. The guided initialization has a length 1, then the behavior is $1$-periodic with a period $6$. Every place has a size one. The only place where one delay occurs is the one between {\em word} and {\em mux} (where a small square appears), however a size one is enough.

As one can see, the AES example is a simply connected graph. In order to run the proposed algorithm, two paths from the sink to each of the sources have been added to the system.

\begin{figure}[tbh]
 \begin{center}
   {\includegraphics[scale=0.57]{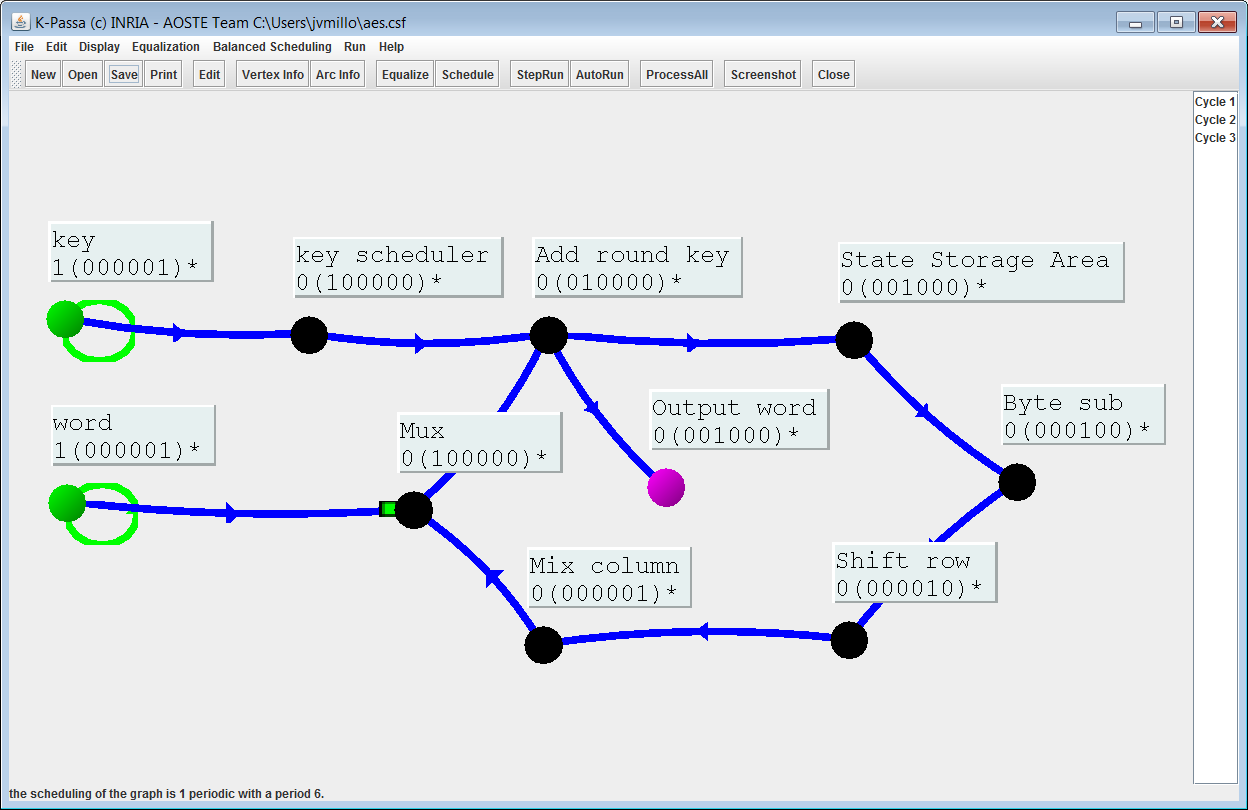}}
  \caption{The MG presents an implementation of the AES encryption standard.} 
  \label{figure_aes}
  \end{center}
\end{figure}

\section{Results and discussion}
\label{section_conclusion}
This paper proposes an algorithm to statically schedule any live and strongly connected MG with a throughput inferior or equals to one. The proposed algorithm computes the balanced ASAP execution where the execution rate is maximal and place sizes are minimal. Moreover, a transformation has been proposed to change a simply connected MG in a strongly connected MG such that the proposed algorithm can be applied.

In the domain to the System-on-Chip design, the proposed algorithm is used to schedule applications which are subject to the problem of long wire latency.
If we compare our approach to the latency insensitive design, this last is not as strict as our approach about the constraint on availability of data on global inputs. It is a purely dynamic solution but the cost for this dynamicity is the duplication of every data path in the circuit and the replacement of every simple register by a two-sized-register to manage the dynamic communication and computation protocol.
This difference makes our approach better for pure data flow system.

\section*{Acknowledgment}
This work has been supported by CIMPACA/SYS2RTL. The authors would like to thanks Benoit Ferrero for his help with the proofs, Anthony Coadou for his constructive remarks, and the anonymous reviewers for their suggestions who have led us in the right direction.

\bibliographystyle{plain}
\bibliography{RR-7891}

\end{document}